\documentclass[11pt]{article}
\usepackage{microtype}
\usepackage{fullpage,amsthm,amssymb,amsmath}
\usepackage{graphicx,algorithmic,algorithm}%, verbatim}
\usepackage{enumerate}
\usepackage{tikz}
\usetikzlibrary{shapes}
\usepackage{hyperref}
\usepackage{complexity}
\usepackage[normalem]{ulem}

\usepackage{xcolor}
\definecolor{dark-red}{rgb}{0.4,0.15,0.15}
\definecolor{dark-blue}{rgb}{0.15,0.15,0.4}
\definecolor{medium-blue}{rgb}{0,0,0.5}
\definecolor{gray}{rgb}{0.5,0.5,0.5}

\hypersetup{
    pdftitle=   {Polynomial-time algorithms for the Longest Induced Path and Induced Disjoint Paths problems on graphs of bounded mim-width},
   pdfauthor=  {Lars Jaffke, O-joung Kwon and Jan Arne Telle}, 
    colorlinks, linkcolor={dark-red},
    citecolor={dark-blue}, urlcolor={medium-blue}
}

\usepackage{authblk}
\newtheorem{theorem}{Theorem}
\newtheorem{lemma}[theorem]{Lemma}
\newtheorem{proposition}[theorem]{Proposition}
\newtheorem{corollary}[theorem]{Corollary}

\newtheorem*{thm:main}{Theorem~\ref{thm:main}}

\newcommand\abs[1]{\lvert #1\rvert}
\newcommand\card[1]{\lvert #1 \rvert}

\newcommand\mimw{\operatorname{mimw}}
\newcommand\mimval{\operatorname{mim}}
\newcommand\mim{\operatorname{mim}}

\newcommand\bd{\operatorname{bd}}

\newcommand\cO{\mathcal O}

\usepackage{xspace}

\newcommand*{\defeq}{\mathrel{\vcenter{\baselineskip0.5ex \lineskiplimit0pt
                     \hbox{\scriptsize.}\hbox{\scriptsize.}}}%
                     =}

\newcommand{\parproblemdef}[4]
{
\begin{quote}
\textsc{#1}\\
\textbf{Input:} #2\\
\textbf{Parameter:} #3\\
\textbf{Question:} #4
\end{quote}
}

\theoremstyle{definition}
\newtheorem{definition}[theorem]{Definition}

\theoremstyle{remark}

\newtheorem{remark}[theorem]{Remark}

\let\plainqed\qedsymbol

\begin{document}

\title{Polynomial-time algorithms for the Longest Induced Path and Induced Disjoint Paths problems on graphs of bounded mim-width\footnote{The work was done while the authors were at Polytechnic University of Valencia, Spain.}}

\author[1]{Lars Jaffke\thanks{Supported by the Bergen Research Foundation (BFS).}}

\author[2]{O-joung Kwon\thanks{Supported by the European Research Council (ERC) under the European Union's Horizon 2020 research and innovation programme (ERC consolidator grant DISTRUCT, agreement No. 648527).}}

\author[1]{Jan Arne Telle}

\affil[1]{Department of Informatics, University of Bergen, Norway. \protect\\ \texttt{\{lars.jaffke, jan.arne.telle\}@uib.no}}
\affil[2]{Logic and Semantics, Technische Universit\"at Berlin, Berlin, Germany. \protect\\ \texttt{ojoungkwon@gmail.com}}

\date\today

\maketitle

\begin{abstract}
	We give the first polynomial-time algorithms on graphs of bounded {\em maximum induced matching width} (mim-width) for problems that are not locally checkable. In particular, we give $n^{\cO(w)}$-time algorithms on graphs of mim-width at most $w$, when given a decomposition, for the following problems: \textsc{Longest Induced Path}, \textsc{Induced Disjoint Paths} and \textsc{$H$-Induced Topological Minor} for fixed $H$. 
	Our results imply that the following graph classes have polynomial-time algorithms for these three problems: \textsc{Interval} and \textsc{Bi-Interval} graphs, \textsc{Circular Arc}, \textsc{Permutation} and \textsc{Circular Permutation} graphs, \textsc{Convex} graphs, \textsc{$k$-Trapezoid},  \textsc{Circular $k$-Trapezoid}, \textsc{$k$-Polygon}, \textsc{Dilworth-$k$} and \textsc{Co-$k$-Degenerate} graphs for fixed $k$.
\end{abstract}

\section{Introduction}
Ever since the definition of the tree-width of graphs emerged from the Graph Minors project of Robertson and Seymour, bounded-width structural graph decompositions have been a successful tool in designing fast algorithms for graph classes on which the corresponding width-measure is small. Over the past few decades, many more width-measures have been introduced, see e.g.~\cite{HOSG2006} for an excellent survey and motivation for width-parameters of graphs. In 2012, Vatshelle~\cite{VatshelleThesis} defined the {\em maximum induced matching width} (mim-width for short) which measures how easy it is to decompose a graph along vertex cuts with bounded maximum induced matching size on the bipartite graph induced by edges crossing the cut. One interesting aspect of this width-measure is that its modeling power is much stronger than tree-width and clique-width and many well-known and deeply studied graph classes such as \textsc{Interval} graphs and \textsc{Permutation} graphs have (linear) mim-width $1$, with decompositions that can be found in polynomial time~\cite{BelmonteV2013,VatshelleThesis}, while their clique-width can be proportional to the square root of the number of vertices. Hence, designing an algorithm for a problem $\Pi$ that runs in $\XP$ time parameterized by mim-width yields polynomial-time algorithms for $\Pi$ on several interesting graph classes at once. 

For \textsc{Locally Checkable Vertex Subset and Vertex Partitioning} (LC-VSVP) problems, a class introduced~\cite{TelleP1997} to capture many well-studied algorithmic problems in a unified framework, Belmonte and Vatshelle~\cite{BelmonteV2013} and Bui-Xuan et al.~\cite{Bui-XuanTV13} provided $\XP$-algorithms on graphs of bounded mim-width. LC-VSVP problems include many NP-hard problems such as \textsc{Maximum Independent Set}, \textsc{Minimum Dominating Set}, and \textsc{$q$-Coloring}. A common feature of these problems is that they (as the name suggests) can be checked locally: Take \textsc{$q$-Coloring} for example. Here, we want to determine whether there is a $q$-partition of the vertex set of an input graph such that each part induces an independent set. The latter property can be checked individually for each vertex by inspecting only its direct neighborhood.

Until now, the only problems known to be $\XP$-time solvable on graphs of bounded mim-width were of the type LC-VSVP. It is therefore natural to ask whether similar results can be shown for problems concerning graph properties that are {\em not} locally checkable. 
In this paper, we mainly study problems related to finding {\em induced} paths in graphs, namely \textsc{Longest Induced Path}, \textsc{Induced Disjoint Paths} and \textsc{$H$-Induced Topological Minor}. 
Although their `non-induced' counterparts are more deeply studied in the literature, also these induced variants have received considerable attention. Below, we briefly survey results for exact algorithms to these three problems on graph classes studied so far.

For the first problem, Gavril~\cite{Gavril2002} showed that \textsc{Longest Induced Path} can be solved in polynomial time for graphs without induced cycles of length at least $q$ for fixed $q$ (the running time was improved by Ishizeki et al.~\cite{Ishizeki2008}), while Kratsch et al.~\cite{KratschMT03} solved the problem on AT-free graphs in polynomial time. Kang et al.~\cite{KangKST2017} recently showed that those classes have unbounded mim-width. However, graphs of bounded mim-width are not necessarily graphs without cycles of length at least $k$ or AT-free graphs.
The second problem derives from the well-known
\textsc{Disjoint Paths} problem which is solvable in $O(n^3)$ time if the number of paths $k$ is a fixed constant, as shown by Robertson and Seymour~\cite{RobertsonS95b}, while if $k$ is part of the input it is NP-complete on graphs of linear mim-width 1 (interval graphs)~\cite{Natarajan}.
In contrast, \textsc{Induced Disjoint Paths} is NP-complete already for $k=2$ paths~\cite{KawarabayashiK2008}.
In this paper we consider the number of paths $k$ as part of the input.
Under this restriction \textsc{Induced Disjoint Paths} is NP-complete on claw-free graphs, as shown by Fiala et al.~\cite{Fiala12}, while Golovach et al.~\cite{GolovachPL16} gave a linear-time algorithm for circular-arc graphs. 
For the third problem, \textsc{$H$-Induced Topological Minor}, we consider $H$ to be a fixed graph.
This problem, and also \textsc{Induced Disjoint Paths}, were both shown solvable in polynomial time on chordal graphs by Belmonte et al.~\cite{Belmonte14}, and on AT-free graphs by Golovach et al.~\cite{GolovachPL12}.

We show that \textsc{Longest Induced Path}, \textsc{Induced Disjoint Paths} and \textsc{$H$-Induced Topological Minor} for fixed $H$ can be solved in time $n^{\cO(w)}$ given a branch decomposition of mim-width $w$. Since bounded mim-width decompositions, usually mim-width 1 or 2, can be computed in polynomial-time for all well-known graph classes having bounded mim-width~\cite{BelmonteV2013}, our results thus provide unified polynomial-time algorithms for these problems on the following classes of graphs: \textsc{Interval} and \textsc{Bi-Interval} graphs, \textsc{Circular Arc}, \textsc{Permutation} and \textsc{Circular Permutation} graphs, \textsc{Convex} graphs, \textsc{$k$-Trapezoid},  \textsc{Circular $k$-Trapezoid}, \textsc{$k$-Polygon}, \textsc{Dilworth-$k$} and \textsc{Co-$k$-Degenerate} graphs for fixed $k$, all graph classes of bounded mim-width~\cite{BelmonteV2013}.\footnote{In~\cite{BelmonteV2013}, results are stated in terms of $d$-neighborhood equivalence, but in the proof, they actually gave a bound on mim-width or linear mim-width. Note that \textsc{Interval}, \textsc{Permutation} and \textsc{$k$-Trapezoid} graphs are AT-free, so polynomial-time algorithms for all three problems were already known~\cite{GolovachPL12}.}

The problem of computing the mim-width of general graphs was shown to be $\W$[1]-hard~\cite{SV16} and no algorithm for computing the mim-width of a graph in $\XP$ time is known. Furthermore, there is no polynomial-time constant-factor approximation for mim-width unless $\NP=\ZPP$~\cite{SV16}.

What makes our algorithms work is an analysis of the structure induced by a solution to the problem on a cut in the branch decomposition. 
% We define branch decompositions in Section~\ref{sec:preliminaries}.
% A cut in a graph $G$ is its vertex partition $(A,B)$. 
There are two ingredients. First, in all the problems we investigate, we are able to show that for each cut induced by an edge of the given branch-decomposition of an input graph, it is sufficient to consider induced subgraphs of size at most $\cO(w)$ as intersections of solutions and the set of edges crossing the cut, where $w$ is the mim-width of the branch decomposition. For instance, in the \textsc{Longest Induced Paths} problem, an induced path is a target solution, and Figure~\ref{fig:lip:solution:intersection} describes a situation where an induced path crosses a cut. We argue that an induced path cannot cross a cut many times if there is no large induced matching between vertex sets $A$ and $B$ of the cut $(A,B)$.  Such an intersection is always an induced disjoint union of paths. Thus, we enumerate all subgraphs of size at most $\cO(w)$, which are induced disjoint unions of paths, and these will be used as indices of our table.

However, a difficulty arises if we recursively ask for a given cut and such an intersection of size at most $\cO(w)$, whether there is an induced disjoint union of paths of certain size in the union of one part and edges crossing the cut, whose intersection on the crossing edges is the given subgraph. The reason is that there are unbounded number of vertices in each part of the cut that are not contained in the given subgraph of size $\mathcal{O}(w)$ but still have neighbors in the other part. We need to control these vertices in such a way that they do not further create an edge in the solution. We control these vertices using vertex covers of the bipartite graph induced by edges crossing the cut. Roughly speaking, if there is a valid partial solution, then there is a vertex cover of such a bipartite graph, which meets all other edges not contained in the given subgraph. The point is that there are only $n^{\mathcal{O}(w)}$ many minimal vertex covers of such a bipartite graph with maximum matching size $w$. We discuss this property in Section~\ref{sec:preliminaries}. Based on these two results, each table will consist of a subgraph of size $\mathcal{O}(w)$ and a vertex cover of the remaining part of the bipartite graph, and we check whether there is a (valid) partial solution to the problem with respect to given information. We can argue that we need to store at most $n^{\cO(w)}$ table entries in the resulting dynamic programming scheme and that each of them can be computed in time $n^{\cO(w)}$ as well.

The strategy for \textsc{Induced Disjoint Paths} is very similar to the one for \textsc{Longest Induced Path}. The only thing to additionally consider is that in the disjoint union of paths, which is guessed as the intersection of a partial solution and edges crossing a cut, we need to remember which path is a subpath of the path connecting a given pair. We lastly provide a one-to-many reduction from \textsc{$H$-Induced Topological Minor} to \textsc{Induced Disjoint Paths}, that runs in polynomial time, and show that it can be solved in time $n^{\mathcal{O}(w)}$. Similar reductions have been shown earlier (see e.g.~\cite{Belmonte14, GolovachPL12}) but we include it here for completeness.

\newcommand{\components}{\mathcal{C}}
\newcommand{\component}[1]{C_{#1}}
\newcommand{\componentG}{\component{G}}	

\section{Preliminaries}\label{sec:preliminaries}
We assume the reader to be familiar with the basic notions in graph theory and parameterized complexity and refer to~\cite{CyganFKLMPPS15,Diestel2010,DowneyF13} for an introduction.

For integers $a$ and $b$ with $a \le b$, we let $[a..b] \defeq \{a, a+1, \ldots, b\}$ and if $a$ is positive, we define $[a] \defeq [1..a]$. 
Throughout the paper, a graph $G$ on vertices $V(G)$ and edges $E(G) \subseteq {V(G) \choose 2}$ is assumed to be finite, undirected and simple. For graphs $G$ and $H$ we say that $G$ is a {\em subgraph} of $H$, if $V(G) \subseteq V(H)$ and $E(G) \subseteq E(H)$ and we write $G \subseteq H$. For a vertex set $X \subseteq V(G)$, we denote by $G[X]$ the subgraph {\em induced} by $X$, i.e.\ $G[X] \defeq (X, E(G) \cap {X \choose 2})$. If $H \subseteq G$ and $X \subseteq V(G)$ then we let $H[X] \defeq H[X \cap V(H)]$.
We use the shorthand $G - X$ for $G[V(G) \setminus X]$. For two (disjoint) vertex sets $X, Y \subseteq V(G)$, we denote by $G[X, Y]$ the bipartite subgraph of $G$ with bipartition $(X, Y)$ such that for $x\in X, y\in Y$, $x$ and $y$ are adjacent in $G$ if and only if they are adjacent in $G[X,Y]$. A {\em cut} of $G$ is a bipartition $(A, B)$ of its vertex set.
For a vertex $v \in V(G)$, we denote by $N[v]$ the set of {\em neighbors} of $v$ in $G$, i.e.\ $N[v] \defeq \{w \in V(G) \mid \{v, w\} \in E(G)\}$.
A set $M$ of edges is a \emph{matching} if no two edges in $M$ share an end vertex, and a matching $\{a_1b_1, \ldots, a_kb_k\}$ is  \emph{induced} if there are no other edges in the subgraph induced by $\{a_1, b_1, \ldots, a_k, b_k\}$. 
For an edge $e =\{v, w\} \in E(G)$, the operation of {\em contracting} $e$ is to remove the edge $e$ from $G$ and identifying $v$ and $w$.

\newcommand*{\dectree}{T}
\newcommand*{\decf}{\mathcal{L}}
\newcommand{\crossinggraph}[1]{G_{#1, \bar{#1}}}
\newcommand{\crossinggraphAB}[2]{G_{#1, #2}}	

\subsection{Branch Decompositions and Mim-Width}
% For a graph $G$ and a vertex set $A \subseteq V(G)$, we denote by $\mimval(A)$ the maximum size of an induced matching in $G[A, V(G) \setminus A]$.
% For a graph $G$ and a vertex set $A \subseteq V(G)$, we define $\mimval_{G}:2^{V(G)}\rightarrow \mathbb{N}$ such that $\mimval_{G}(A)$ is the maximum size of an induced matching of $G[A, V(G)\setminus A]$.
A pair $(\dectree, \decf)$ of a subcubic tree $\dectree$ and a bijection $\decf$ from $V(G)$ to the set of leaves of $\dectree$ is called a \emph{branch decomposition}.
For each edge $e$ of $\dectree$, 
let $\dectree^e_1$ and $\dectree^e_2$ be the two connected components of $\dectree-e$, and 
let $(A^e_1, A^e_2)$ be the vertex bipartition of $G$ such that for each $i\in \{1,2\}$, 
$A^e_i$ is the set of all vertices in $G$ mapped to leaves contained in $\dectree^e_i$ by $\decf$. 
The {\em mim-width of $(\dectree, \decf)$}, denoted by $\mimw(\dectree, \decf)$, is defined as $\max_{e \in E(\dectree)} \mimval(A^e_1)$, where for a vertex set $A \subseteq V(G)$, $\mimval(A)$ denotes the maximum size of an induced matching in $G[A, V(G) \setminus A]$. 
The minimum mim-width over all branch decompositions of $G$ is called the {\em mim-width of $G$}.
% We call $(A^e_1, A^e_2)$ the vertex bipartition of $G$ associated with $e$. 
% For a branch-decomposition $(\dectree, \decf)$ of a graph $G$ and an edge $e$ in $\dectree$ and a symmetric function $f:2^{V(G)}\rightarrow \mathbb{N}$, 
% the \emph{$f$-width} of $e$ is defined as $f(A^e_1)$ where $(A^e_1, A^e_2)$ is the vertex bipartition associated with $e$.
% The \emph{$f$-width} of $(\dectree, \decf)$ is the maximum $f$-width over all edges in $\dectree$, and
% the \emph{$f$-width} of $G$ is the minimum $f$-width over all branch-decompositions of $G$.
If $\abs{V(G)}\le 1$, then $G$ does not admit a branch decomposition, and the mim-width of $G$ is defined to be $0$.
% the \emph{mim-width} of a graph $G$, denoted by $\mimw(G)$, is the $\mimval_G$-width of $G$.

To avoid confusion, we refer to elements in $V(T)$ as {\em nodes} and elements in $V(G)$ as {\em vertices} throughout the rest of the paper.
Given a branch decomposition, one can subdivide an arbitrary edge and let the newly created vertex be the root of $\dectree$, in the following denoted by $r$. Throughout the following we assume that each branch decomposition has a root node of degree two. 
For two nodes $t, t' \in V(T)$, we say that $t'$ is a {\em descendant} of $t$ if $t$ lies on the path from $r$ to $t'$ in $T$.
For $t \in V(\dectree)$, we denote by $G_t$ the subgraph induced by all vertices that are mapped to a leaf that is a descendant of $t$, i.e.\ $G_t \defeq G[X_t]$, where $X_t = \{v \in V(G) \mid \decf^{-1}(t') = v \mbox{ where } t' \mbox{ is a descendant of $t$ in $\dectree$}\}$. We use the shorthand `$V_t$' for `$V(G_t)$' and let $\bar{V_t} \defeq V(G) \setminus V_t$.

The following definitions which we relate to branch decompositions of graphs will play a central role in the design of the algorithms in Section \ref{sec:alg}.

\begin{definition}[Boundary]
	Let $G$ be a graph and $A, B \subseteq V(G)$ such that $A \cap B = \emptyset$. We let $\bd_B(A)$ be the set of vertices in $A$ that have a neighbor in $B$, i.e.\ $\bd_B(A) \defeq \{v \in V(A) \mid N(v) \cap B \neq \emptyset\}$. We define $\bd(A) \defeq \bd_{V(G) \setminus A}(A)$ and call $\bd(A)$ the {\em boundary} of $A$ in $G$.
\end{definition}

\begin{definition}[Crossing Graph]
	Let $G$ be a graph and $A, B \subseteq V(G)$. If $A \cap B = \emptyset$, we define the graph $G_{A, B} \defeq G[\bd_B(A), \bd_A(B)]$ to be the {\em crossing graph from $A$ to $B$}.
\end{definition}
	
	If $(\dectree, \decf)$ is a branch decomposition of $G$ and $t_1, t_2 \in V(\dectree)$ such that the crossing graph $\crossinggraphAB{V_{t_1}}{V_{t_2}}$ is defined, we use the shorthand $\crossinggraphAB{t_1}{t_2} \defeq \crossinggraphAB{V_{t_1}}{V_{t_2}}$. We use the analogous shorthand notations $\crossinggraphAB{t_1}{\bar{t_2}} \defeq \crossinggraphAB{V_{t_1}}{\bar{V_{t_2}}}$ and $\crossinggraphAB{\bar{t_1}}{t_2} \defeq \crossinggraphAB{\bar{V_{t_1}}}{V_{t_2}}$ (whenever these graphs are defined). For the frequently arising case when we consider $\crossinggraph{t}$ for some $t \in V(\dectree)$, we refer to this graph as the {\em crossing graph w.r.t.\ $t$}.
	
	We furthermore use the following notation. Let $G$ be a graph, $v \in V(G)$ and $A \subseteq V(G)$. We denote by $N_A[v]$ the set of neighbors of $v$ in $A$, i.e.\ $N_A[v] \defeq N[v] \cap A$. For $X \subseteq V(G)$, we let $N_A[X] \defeq \bigcup_{v \in X} N_A[v]$. If $(\dectree, \decf)$ is a branch decomposition of $G$ and $t \in V(\dectree)$, we use the shorthand notations $N_t[X] \defeq N_{V_t}[X]$ and $N_{\bar{t}}[X] \defeq N_{\bar{V_t}}[X]$.

\newcommand{\nummis}{\mathrm{mis}}
\newcommand{\nummvc}{\mathrm{mvc}}

\subsection{The Minimal Vertex Covers Lemma} 
Let $G$ be a graph. We now prove that given a set $A \subseteq V(G)$, the number of minimal vertex covers in $G[A, V(G) \setminus A]$ is bounded by $n^{\mimval(A)}$. This observation is crucial to argue that we only need to store $n^{\cO(w)}$ entries at each node in the branch decomposition in all algorithms we design, where $w$ is the mim-width of the given branch decomposition.

Notice that the bound on the number can be easily obtained by combining two results, \cite[Lemma 1]{BelmonteV2013} and \cite[Theorem 3.5.5]{VatshelleThesis}; however, an enumeration algorithm is not given explicitly. To be self-contained, we state and prove it here. 

    \begin{corollary}[Minimal Vertex Covers Lemma]\label{cor:mvc:lemma}
		Let $H$ be a bipartite graph on $n$ vertices with a bipartition $(A,B)$. 
		The number of minimal vertex covers of $H$ is at most $n^{\mimval(A)}$, 
		and the set of all minimal vertex covers of $H$ can be enumerated in time $n^{\mathcal{O}(\mimval(A))}$.
    \end{corollary}
    \begin{proof}
    	Let $w\defeq \mimval(A)$.
    For each vertex set $R\subseteq A$ with $\abs{R}\le w$, 
    let $X_R\subseteq A$ be the set of all vertices having a neighbor in $B\setminus N(R)$.
    We enumerate the sets in 
    \[\mathcal{M}=\{N(R)\cup X_R~\colon R\subseteq A, \abs{R}\le w\}.\]
    Clearly, we can enumerate them in time $n^{\mathcal{O}(w)}$.
    It is not difficult to see that each set in $\mathcal{M}$ is a minimal vertex cover.
    We claim that $\mathcal{M}$ is the set of all minimal vertex covers in $H$.
    
    We use the result by Belmonte and Vatshelle~\cite[Lemma 1]{BelmonteV2013} that
    for a graph $G$ and $A\subseteq V(G)$, $\mim(A)\le k$ if and only if for every $S\subseteq A$, there exists $R\subseteq S$ such that
  $N(R)\cap (V(G)\setminus A)=N(S)\cap (V(G)\setminus A)$ and $\abs{R}\le k$.
    
    Let $U$ be a minimal vertex cover of $H$.
    Clearly, every vertex in $A\setminus U$ has no neighbors in $B\setminus U$, as $U$ is a vertex cover.
    Therefore, by the result of Belmonte and Varshelle, 
    there exists $R\subseteq A\setminus U$ such that $\abs{R}\le w$ and $N(R)\cap B=N(A\setminus U)\cap B=U\cap B$.
    Clearly, $U\cap A=X_R$; if a vertex in $U\cap A$ has no neighbors in $B\setminus U$, then we can remove it from the vertex cover.
    Therefore, $U\in \mathcal{M}$, as required.
    \end{proof}

\section{Algorithms}\label{sec:alg}
In all algorithms presented in this section, we assume that we are given as input an undirected graph $G$ together with a branch decomposition $(\dectree, \decf)$ of $G$ of mim-width $w$, rooted at a degree two vertex obtained from subdividing an arbitrary edge in $\dectree$. We do bottom-up dynamic programming over $(\dectree, \decf)$, starting at the leaves of $\dectree$. 
To obtain our algorithms, we study the structure a solution induces across a cut in the branch decomposition and argue that the size of this structure is bounded by a function only depending on the mim-width. The table entries at each node $t \in V(\dectree)$ are then indexed by all possible such structures and contain the value $1$ if and only if the structure used as the index of this entry constitutes a solution for the respective problem. 
After applying the dynamic programming scheme, the solution to the problem can be obtained by inspecting the table values associated with the root of $\dectree$.

The rest of this section is organized as follows. In Section \ref{sub:sec:alg:LIP} we present an $n^{\cO(w)}$-time algorithm for \textsc{Longest Induced Path}, and in Section \ref{sub:sec:alg:IDP} we give an algorithm for \textsc{Induced Disjoint Paths} with the same asymptotic runtime bound. We give a polynomial-time one-to-many reduction from \textsc{$H$-Induced Topological Minor} (for fixed $H$) to \textsc{Induced Disjoint Paths} in Section \ref{sub:sec:alg:ITM}, yielding an $n^{\cO(w)}$ for the former problem as well.

\subsection{Longest Induced Path}\label{sub:sec:alg:LIP}
For a disjoint union of paths $P$, we refer to its {\em size} as the number of its vertices, i.e.\ $\card{P} \defeq \card{V(P)}$. If $P$ has only one component, we use the terms `size' and `length' interchangeably.
We now give an $n^{\cO(w)}$ time algorithm for the following parameterized problem.
\parproblemdef
	{Longest Induced Path (LIP)/Mim-Width}
	{A graph $G$ with branch decomposition $(\dectree, \decf)$ and an integer $k$}
	{$w \defeq \mimw(\dectree, \decf)$}
	{Does $G$ contain an induced path of length at least $k$?}

% Before we describe the algorithm, we will prove that the size of any partial solution induced by a cut of a branch decomposition of mim-width at most $w$ is linearly bounded in $w$. Note that a partial solution induced by a cut is an induced disjoint union of paths, i.e.\ an induced subgraph that forms a disjoint union of paths.

Before we describe the algorithm, we observe the following. Let $G$ be a graph and $A \subseteq V(G)$ with $\mim(A) = w$ and let $P$ be an induced path in $G$. Then the subgraph induced by edges of $P$ in $\crossinggraph{A}$ and vertices incident with these edges has size linearly bounded by $w$. The following lemma provides a bound of this size.

\begin{lemma}\label{lem:disjoint:paths:size}
	Let $p$ be a positive integer and let $F$ be a disjoint union of paths such that each component of $F$ contains an edge. If $\abs{V(F)}\ge 4p$, then $F$ contains an induced matching of size at least $p$.
	\end{lemma}
	\begin{proof}
	We prove the lemma by induction on $p$. If $p=1$, then it is clear. We may assume $p\ge 2$. Suppose $F$ contains a connected component $C$ with at most $4$ vertices. Then $F-V(C)$ contains at least $4(p-1)$ vertices, and thus it contains an induced matching of size at least $p-1$ by the induction hypothesis. As $C$ contains an edge, $F$ contains an induced matching of size at least $p$. Thus, we may assume that each component of $F$ contains at least $5$ vertices. Let us choose a leaf $v$ of $F$, and let $v_1$ be the neighbor of $v$, and $v_2$ be the neighbor of $v_1$ other than $v$. Since each component of $F-\{v, v_1, v_2\}$ contains at least one edge, we can apply induction to conclude that $F-\{v,v_1,v_2\}$ contains an induced matching of size at least $p-1$. Together with $vv_1$, $F$ contains an induced matching of size at least $p$.
	\end{proof}
	
\begin{remark}
	Unless stated otherwise, any path $P$ (or equivalently, a component of a disjoint union of paths) we refer to throughout the remainder of this section is considered to be nontrivial, i.e.\ $P$ contains at least one edge.
\end{remark}
	
\newcommand*{\dptable}{\mathcal{T}}
\newcommand*{\bdsubset}{S}
\newcommand*{\minvertexcover}{M}
\newcommand*{\matching}{Q}
\newcommand*{\inducedpaths}{\mathcal{I}}
\newcommand*{\lipsolution}{\mathcal{P}}
\newcommand*{\allbdsubsets}{\mathcal{S}}
\newcommand*{\allminvertexcovers}{\mathcal{M}}
\newcommand*{\allmatchings}{\mathcal{Q}}

\newcommand*{\matchingvertex}{s}
\newcommand*{\matchingvertexx}{t}

\newcommand*{\childA}{a}
\newcommand*{\childB}{b}
\newcommand*{\bdsubsett}{R}
\newcommand*{\allbdsubsetts}{\mathcal{R}}
\newcommand*{\indexshort}{\mathfrak{I}}

\newcommand*{\minvertexcoverr}{\minvertexcover'}
\newcommand*{\matchingg}{\matching'}

\newcommand*{\minvertexcoverA}{\minvertexcover^*_\childA}
\newcommand*{\minvertexcoverB}{\minvertexcover^*_\childB}
\newcommand*{\bdsubsettA}{\bdsubsett^*_\childA}
\newcommand*{\bdsubsettB}{\bdsubsett^*_\childB}

\newcommand*{\matchingA}{\matching^*_\childA}
\newcommand*{\matchingB}{\matching^*_\childB}	

\newcommand*{\matchinggraph}{\mathcal{H}}	

\newcommand*{\mIn}[1]{\minvertexcover_{#1}^{\mathrm{in}}}
\newcommand*{\mOut}[1]{\minvertexcover_{#1}^{\mathrm{out}}}
\newcommand{\mOutAB}[2]{\minvertexcover^{\mathrm{out}\left(#2\right)}_{#1}}

% shorthands
\newcommand*{\bds}{\bdsubset}
\newcommand*{\bdss}{\bdsubsett}
\newcommand*{\mvc}{\minvertexcover}
\newcommand*{\mvcc}{\minvertexcoverr}
\newcommand*{\indp}{\inducedpaths}

\newcommand{\crosg}[1]{\crossinggraph{#1}}
\newcommand{\crosgAB}[2]{\crossinggraphAB{#1}{#2}}

\newcommand{\degonevert}{D}
\newcommand{\symdiff}{\triangle}

\newcommand*{\matvtx}{\matchingvertex}
\newcommand*{\matvtxx}{\matchingvertexx}

\begin{figure}
		\centering
		\includegraphics[height=.2\textheight]{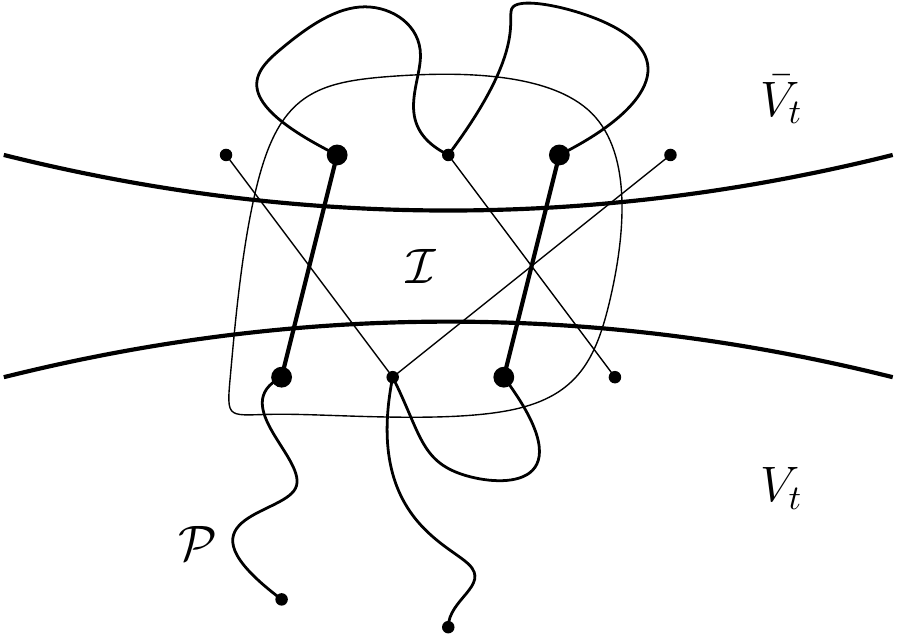}
		\caption{The intersection of an induced path $\mathcal{P}$ with $G[V_t \cup \bd(\bar{V_t})]$, which is an induced disjoint union of paths $\mathcal{I}$. The subgraph $S$ to be used as an index for the corresponding table entry consists of the boldface vertices and edges in $\mathcal{I}$.}
		\label{fig:lip:solution:intersection}
	\end{figure}	

	Before we give the description of the dynamic programming algorithm, we first observe how a solution $\lipsolution$, i.e.\ an induced path in $G$, interacts with the graph $G[V_t \cup \bd(\bar{V_t})]$, for some $t \in V(\dectree)$.
	The intersection of $\lipsolution$ with $G[V_t \cup \bd(\bar{V_t})]$ is an induced disjoint union of paths which we will denote by $\inducedpaths$ in the following.
	To keep the number of possible table entries bounded by $n^{\cO(w)}$, we have to focus on the interaction of $\inducedpaths$ with the crossing graph $\crossinggraph{t}$ w.r.t.\ $t$, in particular the intersection of $\indp$ with its edges. Note that after removing isolated vertices, $\indp$ induces a disjoint union of paths on $\crosg{t}$ which throughout the following we will denote by $\bds$. For an illustration see Figure \ref{fig:lip:solution:intersection}.
	There cannot be any additional edges crossing the cut $(V_t, \bar{V_t})$ between vertices in $\inducedpaths$ on opposite sides of the boundary that are not contained in $V(\bdsubset)$. This property of $\inducedpaths$ can be captured by considering a minimal vertex cover $\minvertexcover$ of the bipartite graph $\crossinggraph{t} - V(\bdsubset)$.
	We remark that the vertices in $\minvertexcover$ play different roles, depending on whether they lie in $\minvertexcover \cap V_t$ or $\minvertexcover \cap \bar{V_t}$. We therefore define the following two sets.
	
\begin{itemize}
	\item $\mIn{t} \defeq \minvertexcover \cap V_t$ is the set of vertices that must be avoided by $\inducedpaths$.
	\item $\mOut{t} \defeq \minvertexcover \cap \bar{V_t}$ is the set of vertices that must be avoided by a partial solution (e.g.\ the intersection of $\lipsolution$ with $G[\bar{V_t}]$) to be combined with $\indp$ to ensure that their combination does not use any edges in $\crosg{t} - V(\bds)$.
\end{itemize}	

Furthermore, $\inducedpaths$ also indicates how the vertices in $\bdsubset[V_t]$ that have degree one in $\bds$ are joined together in the graph $G_t$ (possibly outside $\bd(V_t)$). This gives rise to a collection of vertex pairs $\matching$, which we will refer to as {\em pairings}, with the interpretation that $(\matchingvertex, \matchingvertexx) \in \matching$ if and only if there is a path from $\matchingvertex$ to $\matchingvertexx$ in $\indp[V_t]$.
	
%	Furthermore, $\inducedpaths$ also induces a matching of vertices in $\bdsubset[V_t]$ that have degree one in $\bds$,\footnote{Note that since $\bdsubset$ is considered to be a subgraph of $G$, $\bdsubset[V_t] = G[V(\bdsubset) \cap V_t]$.} indicating how they are joined together in the graph $G_t$ (possibly outside $\bd(V_t)$). 
	
	The description given above immediately tells us how to index the table entries in the dynamic programming table $\dptable$ to keep track of all possible partial solutions in the graph $G[V_t \cup \bd(\bar{V_t})]$: We set the table entry $\dptable[t, (\bdsubset, \minvertexcover, \matching), i, j] = 1$, where $i \in [0..n]$ and $j \in [0..2]$, if and only if the following conditions are satisfied. For an illustration of the table indices, see Figure \ref{fig:lip:mim:table}.
	
	\begin{enumerate}[(i)]
		\item There is a set of induced paths $\inducedpaths$ of total size $i$ in $G[V_t \cup \bd(\bar{V_t})]$ such that $\indp$ has $j$ degree one endpoints in $G_t$. \label{lip:table:def:1}
		\item $E(\inducedpaths) \cap E(\crossinggraph{t}) = E(\bdsubset)$. \label{lip:table:def:2}
		\item $\minvertexcover$ is a minimal vertex cover of $\crossinggraph{t} - V(\bdsubset)$ such that $V(\inducedpaths) \cap \minvertexcover = \emptyset$. (Recall that $\minvertexcover = \mIn{t} \cup \mOut{t}$.) \label{lip:table:def:3}
%		\item $\matching = (\matchingvertex_1, \matchingvertexx_1), \ldots, (\matchingvertex_\ell, \matchingvertexx_\ell)$ is a matching of the vertices in $\bdsubset[V_t]$ that have degree one in $\bds$, 
		\item Let $\degonevert$ denote the vertices in $\bds[V_t]$ that have degree one in $\bds$. Let $\matching = (\matchingvertex_1, \matchingvertexx_1), \ldots, (\matchingvertex_\ell, \matchingvertexx_\ell)$ be a partition of all but $j$ vertices of $\degonevert$ into pairs, throughout the following called a {\em pairing},
		such that if we contract all edges in $\inducedpaths - E(\crosg{t})$ from $\inducedpaths$ incident with at least one vertex not in $S$ (we denote the resulting graph as $\bdsubset \odot \matching$) we obtain the same graph as when adding $\{\matchingvertex_k, \matchingvertexx_k\}$ to $\bdsubset$, for each $k \in [\ell]$. \label{lip:table:def:4}		
	\end{enumerate}

\begin{figure}
	\centering
	\includegraphics[height=.15\textheight]{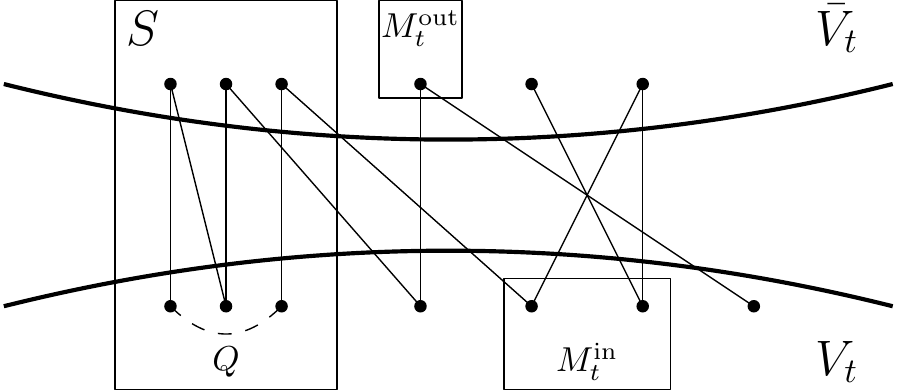}
	\caption{A crossing graph $\crossinggraph{t}$ and the structures associated with the table indices of the algorithm for \textsc{Longest Induced Path}. Note that by (\ref{lip:table:def:1}) and (\ref{lip:table:def:4}) % of the definition of the table entries 
	it follows that if the table entry corresponding to the above structures is $1$, then $j = 0$: Since both degree one endpoints in $S[V_t]$ are paired, this means that the corresponding set of induced paths $\inducedpaths$ has zero degree one endpoints in $G[V_t]$.}
	\label{fig:lip:mim:table}
\end{figure}	

	Regarding (\ref{lip:table:def:4}), observe that $\card{\degonevert} = 2\ell + j$ and that there are $j$ unpaired vertices in $\matching$, each of which is connected to a degree one endpoint of $\indp$ in $G_t$.
	For notational convenience, we will denote by $\dptable_t$ all table entries that have the node $t \in V(\dectree)$ as the first index.

We now show that the solution to \textsc{Longest Induced Path} can be obtained from a table entry corresponding to the root $r$ of $\dectree$ and hence ensure that the information stored in $\dptable$ is sufficient.
\begin{proposition}\label{prop:lip:corr:root}
	$G$ contains an induced path of length $i$ if and only if $\dptable[r, (\emptyset, \emptyset, \emptyset), i, 2] = 1$.
\end{proposition}	
\begin{proof}
	Suppose $G$ contains an induced path $\lipsolution$ of length $i$ and consider the root $r$ of $(\dectree, \decf)$. Clearly, $G[V_r \cup \bd(\bar{V_r})] = G$, since $V_r = V(G)$ and $\bd(\bar{V_r}) = \emptyset$. Together with the fact that $\lipsolution$ is a path (and hence has two degree one endpoints in $G_r = G$, it follows that $\lipsolution$ satisfies Condition (\ref{lip:table:def:1}) for a table entry to be set to $1$. Since $\bd(\bar{V_r}) = \emptyset$, it follows that for any index in $\dptable_r$, $\bdsubset = \emptyset$, $\minvertexcover = \emptyset$ and $\matching = \emptyset$. We can conclude that $\dptable[r, (\emptyset, \emptyset, \emptyset), i, 2] = 1$.
	
	Now suppose $\dptable[r, (\emptyset, \emptyset, \emptyset), i, 2] = 1$. Then there is a set of induced paths $\inducedpaths$ of total size $i$ in $G$ by (\ref{lip:table:def:1}) having two degree one endpoints in $G_r = G$. The latter allows us to conclude that $\inducedpaths$ is in fact a single path.
\end{proof}

Throughout the following, we denote by $\allbdsubsets_t$ the set of all sets of induced disjoint paths in $\crossinggraph{t}$ on at most $4w$ vertices (which includes all possible intersections of partial solutions with $\crossinggraph{t}$ by Lemma \ref{lem:disjoint:paths:size}), for $\bdsubset \in \allbdsubsets_t$ by $\allminvertexcovers_{t, \bdsubset}$ the set of all minimal vertex covers of $\crossinggraph{t} - V(\bdsubset)$ and by $\allmatchings_{t, \bdsubset}$ the set of all pairings of degree one vertices in $\bdsubset[\bd(V_t)]$.
We now argue that the number of such entries is bounded by a polynomial in $n$ whose degree is $\cO(w)$.

\begin{proposition}\label{prop:lip:table:size}
	For each $t \in V(\dectree)$, there are at most $n^{\cO(w)}$ table entries in $\dptable_t$ and they can be enumerated in time $n^{\cO(w)}$.
\end{proposition}
\begin{proof}
	Note that each index is an element of $\allbdsubsets_t \times \allminvertexcovers_{t, \bds} \times \allmatchings_{t, \bds} \times [0..n] \times [0..2]$.
	Since the size of each maximum induced matching in $G_{t, \bar{t}}$ is at most $w$, we know by Lemma \ref{lem:disjoint:paths:size} that the size of each index $\bdsubset$ is bounded by $4w$, so $\card{\allbdsubsets_t} \le \cO(n^{4w})$. By the Minimal Vertex Covers Lemma (Corollary \ref{cor:mvc:lemma}), $\card{\allminvertexcovers_{t, \bds}} \le n^{\cO(w)}$. Since the number of vertices in $\bdsubset$ is bounded by $4w$, we know that $\card{\allmatchings_{t, \bds}} \le w^{\cO(w)}$ and since $i \in [0..n]$ and $j \in [0..2]$, we can conclude that the number of table entries for each $t \in V(\dectree)$ is at most
	\[
		\cO(n^{4w}) \cdot n^{\cO(w)} \cdot w^{\cO(w)} \cdot (n+1) \cdot 3 = n^{\cO(w)},
	\]
	as claimed. Clearly, all elements in $\allbdsubsets_t$ and $\allmatchings_{t, \bds}$ can be enumerated in time $n^{\cO(w)}$ and by the Minimal Vertex Covers Lemma, we know that all elements in $\allminvertexcovers_{t, \bds}$ can be enumerated in time $n^{\cO(w)}$ as well. The claimed time bound on the enumeration of the table indices follows.
\end{proof}

	In the remainder of the proof we will describe how to fill the table entries from the leaves of $T$ to its root, asserting the correctness of the updates in the table. Together with Proposition \ref{prop:lip:corr:root}, this will yield the correctness of the algorithm.
	
	\textbf{Leaves of $\dectree$.} Let $t \in V(\dectree)$ be a leaf node of $\dectree$ and let $v = \decf^{-1}(t)$. We observe that any nonempty intersection of an induced disjoint union of paths in $\crossinggraph{t}$ is a single edge using $v$ or a path of length two whose middle vertex is $v$. Hence, all indices in $\dptable_t$ that are set to $1$ (with nonempty $\bdsubset$) have the following properties: Either $\bdsubset$ is a single edge using $v$, and we denote the set of all such edges by $\allbdsubsets^1_{t, v}$ or a path of length two with $v$ as the middle vertex, and we denote the corresponding set of such paths as $\allbdsubsets^2_{t, v}$. 
%	For each $\bdsubset \in \allbdsubsets^1_{t, v} \cup \allbdsubsets^2_{t, v}$, there is precisely one corresponding minimal vertex cover of $\crossinggraph{t} - V(S)$, namely $N(v) \setminus V(\bdsubset)$. 
	For each $\bds \in \allbdsubsets^1_{t, v} \cup \allbdsubsets^2_{t, v}$, the corresponding minimal vertex cover of $\crosg{t} - V(S)$ is empty, since $v \in V(S)$ and no edges remain in $\crosg{t}$ when we remove $v$.
	Since $v$ is the only vertex in $G_t$, $\matching = \emptyset$ in both of these cases. If $\bdsubset \in \allbdsubsets^1_{t, v}$, then $j = 1$ and if $\bdsubset \in \allbdsubsets^2_{t, v}$ then $j = 0$. Additionally, a table entry is set to $1$ if $\bdsubset = \emptyset$ and $i = 0$ and since the solution is empty in this case, $\matching = \emptyset$ and $j = 0$. The two corresponding minimal vertex covers are $\{v\}$ and $N(v)$.
	Hence, we set the table entries in the leaves as follows. 
	\begin{align*}
		\dptable[t, (\bdsubset, \minvertexcover, \matching), i, j] = 
			\left\{\begin{array}{ll}
				1, &\mbox{if } \bdsubset \in \allbdsubsets^1_{t, v}, \minvertexcover = \emptyset, \matching = \emptyset, i = 2 \mbox{ and } j = 1 \\
				1, &\mbox{if } \bdsubset \in \allbdsubsets^2_{t, v}, \minvertexcover = \emptyset, \matching = \emptyset, i = 3 \mbox{ and } j = 0 \\
				1, &\mbox{if } \bdsubset = \emptyset, \minvertexcover \in \{\{v\}, N(v)\}, \matching = \emptyset, i = 0 \mbox{ and } j = 0 \\
				0, &\mbox{otherwise}
			\end{array}\right.		
	\end{align*}
	
	\textbf{Internal nodes of $\dectree$.} Let $t \in V(\dectree)$ be an internal node of $\dectree$, let $(\bdsubset, \minvertexcover, \matching) \in \allbdsubsets_t \times \allminvertexcovers_{t, \bdsubset} \times \allmatchings_{t, \bdsubset}$, let $i \in [0..n]$ and $j \in [0..2]$. We show how to compute the table entry $\dptable[t, (\bdsubset, \minvertexcover, \matching), i, j]$ from table entries corresponding to the children $\childA$ and $\childB$ of $t$ in $\dectree$. To do so, we have to take into account the ways in which partial solutions for $G[V_\childA \cup \bd(\bar{V_\childA})]$ and $G[V_\childB \cup \bd(\bar{V_\childB})]$ interact. We therefore try all pairs of indices $\indexshort_\childA = ((\bdsubset_\childA, \minvertexcover_\childA, \matching_\childA), i_\childA, j_\childA)$, $\indexshort_\childB = ((\bdsubset_\childB, \minvertexcover_\childB, \matching_\childB), i_\childB, j_\childB)$ and for each such pair, first check whether it is `compatible' with $\indexshort_t$: We say that $\indexshort_\childA$ and $\indexshort_\childB$ are {\em compatible with $\indexshort_t$} if and only if any partial solution $\inducedpaths_\childA$ represented by $\indexshort_\childA$ for $G[V_\childA \cup \bd(\bar{V_\childA})]$ and $\inducedpaths_\childB$ represented by $\indexshort_\childB$ for $G[V_\childB \cup \bd(\bar{V_\childB})]$ can be combined to a partial solution $\inducedpaths_t$ for $G[V_t \cup \bd(\bar{V_t})]$ that is represented by the index $\indexshort_t$. We then set $\dptable_t[\indexshort_t] \defeq 1$ if and only if we can find a compatible pair of indices $\indexshort_\childA$, $\indexshort_\childB$ as above such that $\dptable_\childA[\indexshort_\childA] = 1$ and $\dptable_\childB[\indexshort_\childB] = 1$.

\begin{description}
	\item[Step 0 (Valid Index).] We first check whether the index $\indexshort_t$ can represent a valid partial solution of $G[V_t \cup \bd(\bar{V_t})]$. The definition of the table entries requires that $\bdsubset \odot \matching$ is a disjoint union of paths, so if $\bdsubset \odot \matching$ is not a disjoint union of paths, we set $\dptable_t[\indexshort_t] \defeq 0$ and skip the remaining steps. In general, the number of degree one vertices in $V(\bdsubset \odot \matching) \cap V_t$ has to be equal to $j$ and we can proceed as described in Steps 1-4, except for the following special cases.
	\begin{description}
		\item[Special Case 1 ($j = 2$, $V(\bdsubset \odot \matching) \cap V_t$ has $0$ deg.\ $1$ vertices).] This is the case when $\inducedpaths$ does not contain any edge in $E(\crossinggraph{t})$. It immediately follows that $\bdsubset$ has to be empty (and hence $\matching$ has to be empty). If not, we set $\dptable_t[\indexshort_t] = 0$ and skip the remaining steps. We would like to remark that the case when $j=2$ and $\bdsubset = \emptyset$ will have to be dealt with separately in Step 3, since $\bdsubset \odot \matching$ is the empty graph.
		\item[Special Case 2 ($j = 2$, $V(\bdsubset \odot \matching) \cap V_t$ has $2$ deg.\ $1$ vertices in same component).] In this case, $j = 2$ and the degree one vertices in $V(\bdsubset \odot \matching) \cap V_t$ are in the same component $C$ of $\bdsubset \odot \matching$. 
		Then, $\bdsubset \odot \matching$ has to consist of a single component (and $\inducedpaths$ of a single path): If there was another component $C'$ in $\bdsubset \odot \matching$, $C$ and $C'$ could never be joined together to become a single path in the vertices of $\bar{V_t}$ and hence we can never obtain a valid solution from the partial solution represented by this index.
		So if $\bdsubset \odot \matching$ has more than one component, we set $\dptable_t[\indexshort_t] = 0$ and skip the remaining steps, otherwise we do not have to pay any further attention to this case in the following computations.
		\item[Special Case 3 ($j = 0$ and $S = \emptyset$).] This is the case when $\inducedpaths$ does not use any vertex of $V_t$. We then set $\dptable_t[\indexshort_t] = 1$ if and only if $i = 0$ and skip the remaining steps.
	\end{description} 
	
	\item[Step 1 (Induced disjoint unions of paths).] We now check whether $\bdsubset_\childA$ and $\bdsubset_\childB$ are compatible with $\bdsubset$. We have to ensure that
	\begin{itemize}
		\item $\bdsubset \cap \crossinggraphAB{\childA}{\bar{t}} = \bdsubset_\childA \cap \crossinggraphAB{\childA}{\bar{t}}$,
		\item $\bdsubset \cap \crossinggraphAB{\childB}{\bar{t}} = \bdsubset_\childB \cap \crossinggraphAB{\childB}{\bar{t}}$ and
		\item $\bdsubset_\childA \cap \crossinggraphAB{\childA}{\childB} = \bdsubset_\childB \cap \crossinggraphAB{\childA}{\childB}$.
	\end{itemize}
	If these conditions are not satisfied, we skip the current pair of indices $\indexshort_\childA$, $\indexshort_\childB$.
	In the following, we use the notation $\bdsubsett = \bdsubset_\childA \cap \crossinggraphAB{\childA}{\childB} ~(= \bdsubset_\childB \cap \crossinggraphAB{\childA}{\childB})$.
	\item[Step 2 (Pairings of degree one vertices and $j$).] 
	First, we deal with Special Case 1, i.e.\ $j=2$ and $\bdsubset = \emptyset$.
	We then check whether the graph obtained from taking $\bdsubsett$ and adding an edge (and, if not already present, the corresponding vertices) for each pair in $\matching_\childA$ and $\matching_\childB$ is a single induced path. Note that we require the values of the integers $j_\childA$ and $j_\childB$ to be the number of endpoints of the resulting path in $V_\childA$ and $V_\childB$, respectively.
	
	Since the case $j = 0$ and $S = \emptyset$ is dealt with in Special Case 3 and $S$ cannot be empty whenever $j =1$, we may from now on assume that $\bdsubset \neq \emptyset$ and hence $\bdsubset \odot \matching \neq \emptyset$.\footnote{Note that it could still happen that $\matching = \emptyset$ but since this does not essentially influence the following argument, we assume that $\matching \neq \emptyset$.}
	
	Consider the graph on vertex set $V(\bdsubset) \cup V(\bdsubsett)$ whose edges consist of the edges in $\bdsubset$ and $\bdsubsett$ together with the pairs in $\matching_\childA$ and $\matching_\childB$. We then contract all edges in $\bdsubsett$ and all edges that were added due to the pairings $\matching_\childA$ and $\matching_\childB$ and incident with a vertex not in $S$, and denote the resulting graph by $\matchinggraph$. Then, $\matching_\childA$ and $\matching_\childB$ are compatible if and only if $\matchinggraph = \bdsubset \odot \matching$. By the definition of the table entries (and since by Step 0, $\bdsubset \odot \matching$ is a disjoint union of paths) we can then see that $\matching_\childA, \matching_\childB$ together with the edges of $\bdsubsett$ connect the paired degree one vertices of $\matching$ as required. 
	We furthermore need to ensure that the values of the integers $j_\childA$ and $j_\childB$ are the number of degree one endpoints in $\matchinggraph[V_\childA]$ and $\matchinggraph[V_\childB]$, respectively.	
	For an illustration see Figure \ref{fig:lip:mim:join:matching}.
	
	\begin{figure}
		\centering
		\includegraphics[height=.22\textheight]{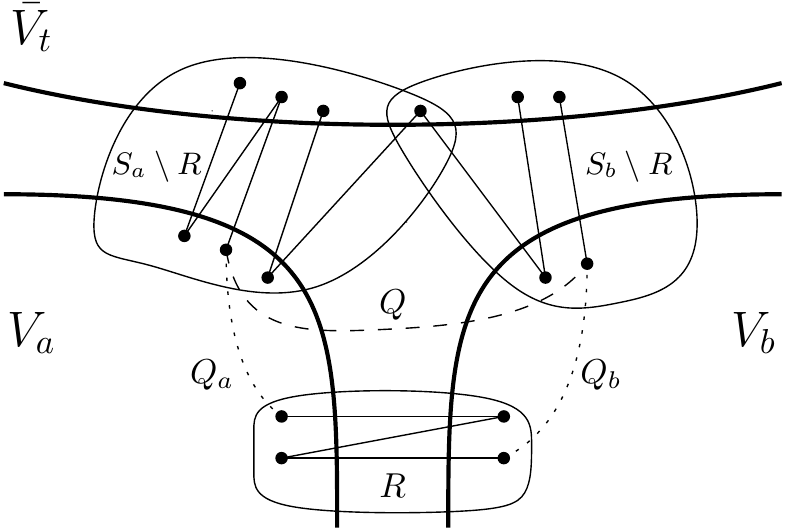}
		\caption{Step 3 of the join operation. Recall that $S_a \cap \crosgAB{\childA}{\childB} = S_b \cap \crosgAB{\childA}{\childB} = R$ by Step 1.}
		\label{fig:lip:mim:join:matching}
	\end{figure}	
		
	\item[Step 3 (Minimal vertex covers).] We now describe the checks we have to perform to ensure that $\minvertexcover_\childA$ and $\minvertexcover_\childB$ are compatible with $\minvertexcover$, which from now on we will denote by $\mvc_t$ to avoid confusion. For ease of exposition, we denote by $\inducedpaths_t$, $\inducedpaths_\childA$ and $\inducedpaths_\childB$ (potential) partial solutions corresponding to $\indexshort_t$, $\indexshort_\childA$ and $\indexshort_\childB$, respectively.
	
	Recall that the purpose of the minimal vertex cover $\mvc_t$ is to ensure that no unwanted edges appear between vertices used by the partial solution $\indp_t$ and any partial solution of $G[\bar{V_t} \setminus \bd(\bar{V_t})]$ that can be combined with $\indp_t$. Hence, when checking whether $\indp_\childA$ and $\indp_\childB$ can be combined to $\indp_t$ without explicitly having access to these sets of induced disjoint paths, we have to make sure that the indices $\indexshort_\childA$ and $\indexshort_\childB$ assert the absence of unwanted edges --- for any intersection of a partial solution with $\crosg{\childA}$ and $\crosg{\childB}$, as well as with $\crosgAB{\childA}{\childB}$. Recall that $\crosg{\childA} = \crosgAB{\childA}{\bar{t}} \cup \crosgAB{\childA}{\childB}$ and $\crosg{\childB} = \crosgAB{\childB}{\bar{t}} \cup \crosgAB{\childA}{\childB}$.
	
	We distinguish several cases, depending on where the unwanted edge might appear: First, between two intermediate vertices of partial solutions and second, between a vertex in $\bds_\childA$ or $\bds_\childB$ and an intermediate vertex. Step 3.1 handles the former and Step 3.2 the latter. In Step 3.1, we additionally have to distinguish whether the edge might appear in $\crosgAB{\childA}{\childB}$ or in $\crosgAB{\childA}{\bar{t}}$ (respectively, in $\crosgAB{\childB}{\bar{t}}$).  
	
	In the following, we let $\mOutAB{\childA}{\childB} \defeq \mOut{\childA} \cap V_\childB$ and $\mOutAB{\childB}{\childA} \defeq \mOut{\childB} \cap V_\childA$.
	\begin{description}
		\item[Step 3.1.1 (intermediate-intermediate, $\crosgAB{\childA}{\childB}$).]	$\mOutAB{\childA}{\childB} \subseteq \mIn{\childB}$ and $\mOutAB{\childB}{\childA} \subseteq \mIn{\childA}$.
		 
		 A vertex $v \in \mOutAB{\childA}{\childB}$ can have a neighbor $w \in V_\childA$ which is used as an intermediate vertex in $\inducedpaths_\childA$. Hence, to avoid that the unwanted edge $\{v, w\}$ appears in the combined solution $\inducedpaths_\childA \cup \inducedpaths_\childB$, we have to make sure that $v$ is not used by $\inducedpaths_\childB$, which is asserted if $v \in \mIn{\childB}$. By a symmetric argument we justify that $\mOutAB{\childB}{\childA} \subseteq \mIn{\childA}$.
	
	\item[Step 3.1.2 (intermediate-intermediate, $\crosgAB{\childA}{\bar{t}}$ or $\crosgAB{\childB}{\bar{t}}$).] We have to check the following two conditions, the first one regarding $\mIn{t}$ and the second one regarding $\mOut{t}$.
	\begin{enumerate}[(a)]
		\item $\mIn{t} \subseteq \mIn{\childA} \cup \mIn{\childB}$: By the definition of $\mIn{t}$, $\inducedpaths_t$ has to avoid the vertices in $\mIn{t}$. Hence, $\inducedpaths_\childA$ and $\inducedpaths_\childB$ have to avoid the vertices in $\mIn{t}$ as well, which is ensured if for $v \in \mIn{t} \cap V_\childA$, we have that $v \in \mIn{\childA}$ and for $w \in \mIn{t} \cap V_\childB$, we have that $w \in \mIn{\childB}$.
		\item For each vertex $v \in \mOut{t}$ having a neighbor $x$ in $V_\childA$ such that $x$ is also contained in $V_\childA \setminus (V(\bdsubset_\childA) \cup \mIn{\childA})$, we have that $v \in \mOut{\childA}$. Recall that by the definition of $\mOut{t}$, $\inducedpaths_t$ could use the vertex $x$ as an intermediate vertex. If $x \notin V(\bdsubset_\childA) \cup \mIn{\childA}$, this means that $x$ might be used by $\inducedpaths_\childA$ as an intermediate vertex as well. Now, in a table entry representing a partial solution $\inducedpaths_\childA$ using $x$, this is signalized by having $v \in \mOut{\childA}$.
		 We check the analogous condition for $\mOut{\childB}$.
	\end{enumerate}
		 \item[Step 3.2 (intermediate-($\bds_\childA$ or $\bds_\childB$)).] $N_\childA[V(\bds_\childB) \setminus V(\bds_\childA)] \subseteq \mIn{\childA}$ and $N_\childB[V(\bds_\childA) \setminus V(\bds_\childB)] \subseteq \mIn{\childB}$.
		 
		 We justify the first condition and note that the second one can be argued for symmetrically. Clearly, $\indp_\childA$ cannot have a neighbor $x$ of any vertex $v \in V(\bds_\childB)$ as an intermediate vertex, if $\indp_\childA$ is to be combined with $\indp_\childB$. However, if $v \in V(\bds_\childA)$, then $\indp_\childA$ does not use $x$ by Part (\ref{lip:table:def:2}) of the definition of the table entries. Note that this includes all vertices in $V(\bdss) \subseteq V(\bds_\childA)$. If on the other hand, $v \in V(\bds_\childB) \setminus V(\bds_\childA)$ then the neighbors of $v$ have not been accounted for earlier, since $v$ is not a vertex in the partial solution $\indp_\childA$. Hence, we now have to assert that $\indp_\childA$ does not use $x$, the neighbor of $v$, and so we require that $x \in \mIn{\childA}$.
	\end{description}
	
	\item[Step 4 ($i$).] We consider all pairs of integers $i_\childA$, $i_\childB$ such that $i = i_\childA + i_\childB - \card{V(\bdsubsett)}$. By Step 2, all vertices in $R$ are used in the partial solution $\indp_t$. They are counted twice, since they are both accounted for in $\indp_\childA$ and in $\indp_\childB$.

\end{description}

Now, we let $\dptable_t[\indexshort_t] = 1$ if and only if there is a pair of indices $\indexshort_\childA = ((\bds_\childA, \mvc_\childA, \matching_\childA), i_\childA, j_\childA)$ and $\indexshort_\childB = ((\bds_\childB, \mvc_\childB, \matching_\childB), i_\childB, j_\childB)$ passing all checks performed in Steps 1-4 above, such that $\dptable_\childA[\indexshort_\childA] = 1$ and $\dptable_\childB[\indexshort_\childB] = 1$. This finishes the description of the algorithm.

\begin{proposition}\label{prop:lip:corr:join}
	Let $t \in V(\dectree)$. The table entries $\dptable_t[\indexshort_t]$ computed according to Steps 0-4 above are correct.
\end{proposition}
\begin{proof}
	Suppose there is partial solution $\indp_t$, an induced disjoint union of paths in $G[V_t \cup \bd(\bar{V_t})]$. We claim that for any index $\indexshort_t$ representing the partial solution $\indp_t$, $\dptable_t[\indexshort_t] = 1$ after the join operation described by Steps 0-4, assuming as the induction hypothesis that the table values of the children $\childA$ and $\childB$ of $t$ are computed correctly.
	Any index $\indexshort_t$ representing $\indp_t$ has the following properties: $\bds = \indp_t \cap \crosg{t}$ and the pairing $\matching$ of the degree one vertices in $\bds[V_t]$ can be obtained by letting $(\matvtx, \matvtxx) \in \matching$ if and only if there is a path connecting $\matvtx$ and $\matvtxx$ in $\indp_t[V_t]$. We furthermore have that $i = |V(\indp_t)|$ and $j$ is the number of degree one endpoints in $\indp_t[V_t \setminus \bd(V_t)]$. However, there might be several choices for the minimal vertex cover $\mvc$ of $\crosg{t} - V(\bds)$. By the definition of the table entries, $\mIn{t} \subseteq V_t \setminus V(\indp_t)$ and $\mOut{t}$ contains all vertices of $\bar{V_t}$ that have a neighbor in $V(\indp_t) \setminus V(\bds)$.
	Throughout the remainder of the proof, we fix one such vertex cover $\mvc$. Clearly, $\indexshort_t = ((\bds, \mvc, \matching), i, j)$ is a valid index according to the checks done in Step 0 and represents $\indp_t$ in the sense of the definition of the table entries.

	We observe that $\indp_t$ induces partial solutions for the children $\childA$ and $\childB$ of $t$: We let $\indp_\childA = \indp_t \cap G[V_\childA \cup \bd(\bar{V_\childA})]$ and $\indp_\childB = \indp_t \cap G[V_\childB \cup \bd(\bar{V_\childB})]$. We now show how to obtain indices $\indexshort_\childA$ and $\indexshort_\childB$ representing $\indp_\childA$ and $\indp_\childB$, respectively, that are compatible to be combined to the index $\indexshort_t$ in the sense of Steps 1-4 of the algorithm presented above. In complete analogy to above, we obtain $\bds_\childA = \indp_\childA \cap \crosg{\childA}$, the pairing $\matching_\childA$ of degree one endpoints in $\bds_\childA[V_\childA]$, and the integers $i_\childA$ and $j_\childA$ and we proceed in the same way to obtain $\bds_\childB, \matching_\childB, i_\childB$ and $j_\childB$. Again, there will be several choices for the minimal vertex covers $\mvc_\childA$ of $\crosg{\childA} - V(\bds_\childA)$ and $\mvc_\childB$ of $\crosg{\childB} - V(\bds_\childB)$, of which we will choose one `representative'. For the details see further below. What we can immediately verify is that $\bds_\childA$ and $\bds_\childB$ are compatible with $\bds$ in the sense of Step 1, that $\matching_\childA$ and $\matching_\childB$ are compatible with $\matching$ in the sense of Step 2, that the values of $j_\childA$ and $j_\childB$ are compatible with $j$, and that $i_\childA$ and $i_\childB$ are compatible with $i$ in accordance with Step 4. Throughout the following, we denote by $\bdss = \bds_\childA \cap \bds_\childB$ and note that $\bdss$ is contained in the crossing graph $\crosgAB{\childA}{\childB}$.
	
	We now explain how to construct a pair of minimal vertex covers $\mvc_\childA$ and $\mvc_\childB$ of $\crosg{\childA} - V(\bds_\childA)$ and $\crosg{\childB} - V(\bds_\childB)$, respectively, making sure that they are compatible with $\mvc$ in the sense of Step 3 of the algorithm description. 
	Note that in the following, we only show how to construct $\mvc_\childA$ and we perform the symmetric steps to construct $\mvc_\childB$. Our strategy is as follows: We add vertices to $\mvc_\childA$ in consecutive stages and ensure in each stage that all vertices we add cover an edge that has not been covered by $\mvc_\childA$ so far. Hence minimality of the resulting vertex cover will follow. We furthermore point out at each stage, which part of Step 3 in the description of the join operation it satisfies.
	\begin{enumerate}[1.]
		\item Let $w$ be an intermediate vertex of $\indp_\childB$ and let $x$ be a neighbor of $w$ in $V_\childA$. Then, add $x$ to $\mIn{\childA}$. Now, let $w$ be an intermediate vertex of $\indp_\childA$ and $x$ be a neighbor of $\indp_\childA$ in $\bar{V_\childA}$. Then, add $x$ to $\mOut{\childA}$. In both cases, this covers the edge $\{w, x\}$. Since we do not add any more vertices to $\mOutAB{\childA}{\childB}$ in the remaining construction, the vertex $w$ in the first case will never be added to $\mvc_\childA$ and since in the second case, $w$ is an intermediate vertex, it will not be added to $\mvc_\childA$ either. Hence this stage of the construction cannot violate the minimality condition of $\mvc_\childA$. 
		% and since $w$ is an intermediate vertex, it will never be added to $\mvc_\childA$. This satisfies Part (\ref{lip:table:def:3}) of the table definition. 
		Note that since we proceed symmetrically for $\mvc_\childB$, the condition in Step 3.1.1 is also satisfied by this part of the construction.\label{lip:corr:mvc:1}
		\item We add any $x \in \mIn{t} \cap V_\childA$ to $\mIn{\childA}$. Since $\mvc_t$ is minimal, $x$ covers an edge $\{x, w\}$ in $\crosgAB{\childA}{\bar{t}}$, where $w \in \bar{V_t} \setminus V(\bds_t)$. Hence, the edge $\{x, w\}$ has not been covered so far by $\mvc_\childA$. This ensures that the condition in Step 3.1.2 a) is met. \label{lip:corr:mvc:2}
		\item If a vertex $x$ in $V_\childA \setminus V(\bds_\childA)$ has a neighbor $w$ in $V(\bds_\childB) \setminus V(\bds_\childA)$, then add $x$ to $\mIn{\childA}$, so $x$ covers the edge $\{x, w\}$. Note that this vertex $w$ is different from the one in Stage \ref{lip:corr:mvc:2}, as $w \in V(\bds_\childB) \setminus V(\bds_\childA)$ and hence $w$ is a vertex of $\bds_t$ or it is {\em not} contained in $\bar{V_t}$. This asserts that the condition in Step 3.2 is satisfied.
		\item Consider the subgraph $G^*$ of $\crosg{\childA} - V(\bds_\childA)$ on the edges that have not been covered by $\mvc_\childA$ so far. By Stage \ref{lip:corr:mvc:1}, this graph does not touch any vertex in $\indp_\childA$ or $\indp_\childB$. We then add all vertices in $V(G^*) \cap V_\childA$ to $\mIn{\childA}$. This ensures that there is no vertex $v \in \mOut{t}$ that violates the condition of Step 3.1.2 b). \label{lip:corr:mvc:4}
	\end{enumerate}
	
	By the above construction, in particular by Stage \ref{lip:corr:mvc:4}, we verify that $\mvc_\childA$ is a vertex cover of $\crosg{\childA} - V(\bds_\childA)$ and in each stage, we checked that $\mvc_\childA$ remains minimal after adding the corresponding vertices.
	
	We have shown how to derive from the partial solutions $\indp_\childA$ and $\indp_\childB$ a pair of table entries $\indexshort_\childA$, $\indexshort_\childB$ that represent them. Since we assume inductively that the algorithm is correct for the children of $t$, we know that $\dptable_\childA[\indexshort_\childA] = 1$ and $\dptable_\childB[\indexshort_\childB] = 1$.  By the description of the algorithm, this implies that $\dptable_t[\indexshort_t] = 1$ which concludes this direction of the proof.
	
	For the other direction, suppose there exists an index $\indexshort_t$ such that $\dptable_t[\indexshort_t] = 1$. By the description of the algorithm we can find pairs of indices $\indexshort_\childA$ and $\indexshort_\childB$ such that $\dptable_\childA[\indexshort_\childA] = 1$ and $\dptable_\childB[\indexshort_\childB] = 1$. Assuming for the induction hypothesis that the values of the children of $t$ are computed correctly, we obtain the corresponding partial solutions $\indp_\childA$ and $\indp_\childB$. Following the description of the algorithm, we can then see that $\indp_\childA$ and $\indp_\childB$ can be combined to a valid solution $\indp_t$ which is represented by $\indexshort_t$.
\end{proof}

By Propositions \ref{prop:lip:corr:root} and \ref{prop:lip:corr:join} and the fact that in the leaf nodes of $\dectree$, we enumerate all possible partial solutions, we know that the algorithm we described is correct.
Since by Proposition \ref{prop:lip:table:size}, there are at most $n^{\cO(w)}$ table entries at each node of $\dectree$ (and they can be enumerated in time $n^{\cO(w)}$), the value of each table entry in $\dptable_t$ as above can be computed in time $n^{\cO(w)} \cdot n^{\cO(w)} \cdot n^{\cO(1)} = n^{\cO(w)}$, since each check described in Steps 0-4 can be done in time polynomial in $n$. Since additionally, $|V(\dectree)| = \cO(n)$, the total runtime of the algorithm is $n^{\cO(w)} \cdot n^{\cO(w)} \cdot \cO(n) = n^{\cO(w)}$ and we have the following theorem.

\begin{theorem}
	There is an algorithm that given a graph $G$ on $n$ vertices and a branch decomposition $(\dectree, \decf)$ of $G$, solves \textsc{Longest Induced Path} in time $n^{\cO(w)}$, where $w$ denotes the mim-width of $(\dectree, \decf)$.
\end{theorem}

\newcommand{\terminal}{x}
\newcommand{\terminall}{y}
\newcommand{\terminals}{X}

\newcommand*{\idpsolution}{\mathcal{P}}
\newcommand*{\idpsol}{\idpsolution}

\newcommand*{\labf}{\lambda}
\newcommand*{\alllabf}{\Lambda}

\newcommand{\matchinglab}[1]{Q_{#1}}
\newcommand{\matlab}{\matchinglab{\labf}}

\subsection{Induced Disjoint Paths}\label{sub:sec:alg:IDP}

In this section, we build upon the ideas of the algorithm for \textsc{Longest Induced Path} presented above to obtain an $n^{\cO(w)}$-time algorithm for the following parameterized problem.

\parproblemdef
	{Induced Disjoint Paths (IDP)/Mim-Width}
	{A graph $G$ with branch decomposition $(\dectree, \decf)$ and pairs of vertices $(\terminal_1, \terminall_1)$, $\ldots$, $(\terminal_k, \terminall_k)$ of $G$.}
	{$w \defeq \mimw(\dectree, \decf)$}
	{Does $G$ contain a set of vertex-disjoint induced paths $P_1,\ldots,P_k$, such that for $i \in [k]$, $P_i$ is a path from $\terminal_i$ to $\terminall_i$ and for $i \neq j$, $P_i$ does not contain a vertex adjacent to a vertex in $P_j$?}
	
	Throughout the remainder of this section, we refer to the vertices $\{\terminal_i, \terminall_i\}$, where $i \in [k]$ as the {\em terminals} and we denote the set of all terminals by $\terminals \defeq \bigcup_{i \in [k]} \{\terminal_i, \terminall_i\}$. We furthermore use the following notation: We denote by $\components(G)$ the set of all connected components of $G$ and for a vertex $v \in V(G)$, $\component{G}(v)$ refers to the connected component containing $v$.

We observe how a solution $\idpsol = (\idpsol_1, \ldots, \idpsol_k)$ interacts with the graph $G[V_t \cup \bd(\bar{V_t})]$, for some $t \in V(\dectree)$. In this case, for each $i \in [k]$, $\idpsol_i$ is an $(\terminal_i, \terminall_i)$-path and additionally for $j \neq i$, there is no vertex in $\idpsol_i$ adjacent to a vertex of $\idpsol_j$. The intersection of $\idpsol$ with $G[V_t \cup \bd(\bar{V_t})]$ is a subgraph $\indp = (\indp_1,\ldots, \indp_k)$, where each $\indp_i$ is a (possibly empty) induced disjoint union of paths which is the intersection of the $(\terminal_i, \terminall_i)$-path $\idpsol_i$ with $G[V_t \cup \bd(\bar{V_t})]$. Note that each terminal $v_i \in \{\terminal_i, \terminall_i\}$ that is contained in $V_t \cup \bd(\bar{V_t})$ is also contained in $V(\indp_i)$.

	Again our goal is to bound the number of table entries at each node $t \in V(\dectree)$ by $n^{\cO(w)}$, so we focus on the intersection of $\indp$ with the crossing graph $\crosg{t}$. % As in the case of the \textsc{Longest Induced Path} problem, this intersection is an induced disjoint union of paths, so the size of this intersection is again bounded by $\cO(w)$ by Lemma \ref{lem:disjoint:paths:size}. This fact which is again crucially used to bound the number of table entries by $n^{\cO(w)}$.
	There are several reasons why $\indp_i$ can have a nonempty intersection with the crossing graph $\crosg{t}$: If precisely one of $\terminal_i$ and $\terminall_i$ is contained in $V_t$, then the path $\idpsol_i$ must cross the boundary of $G_t$.
If both $\terminal_i$ and $\terminall_i$ are contained in $V_t$ ($\bar{V_t}$), yet $\idpsol_i$ uses a vertex of $\bar{V_t}$ ($V_t$), then it crosses the boundary of $G_t$. % Note however that by the observation in the previous paragraph, at most $\cO(w)$ paths can cross any given cut.
	
	We now turn to the definition of the table indices. Let us first point out what table indices in the resulting algorithm for \textsc{Induced Disjoint Paths} have in common with the indices in the algorithm for \textsc{Longest Induced Path} and we refer to Section \ref{sub:sec:alg:LIP} for the motivation and details. These similarities arise since in both problems, the intersection of a solution with a crossing graph $\crosg{t}$ is an induced disjoint union of paths.
	
	\begin{itemize}
		\item The intersection of $\indp$ with the edges of $\crosg{t}$ is $\bds$, an induced disjoint union of paths where each component contains at least one edge.
		\item $\mvc$ is a minimal vertex cover of $\crosg{t} - V(\bds)$ such that $\mvc \cap V(\bds) = \emptyset$.
	\end{itemize}
	
	The first important observation to be made is that by Lemma \ref{lem:disjoint:paths:size}, the number of components of $\bds$ is linearly bounded in $w$ and hence at most $\cO(w)$ paths of $\idpsol$ can have a nonempty intersection with $\crosg{t}$. We need to store information about which path $\idpsol_i$ (resp., to which $\indp_i$) the components of $\bds$ correspond to. To do so, another part of the index will be a labeling function $\labf \colon \components(\bds) \to [k]$, whose purpose is to indicate that each component $C \in \components(\bds)$ is contained in $\indp_{\labf(C)}$. Remark that we just observed that each such $\labf$ contains at most $\cO(w)$ entries.
	
	Let $i \in [k]$. Again, we need to indicate how (some of) the components of $\bds$ are connected via $\indp_i$ in $G[V_t]$. As before, we do so by considering a pairing of the vertices in $S[V_t]$ that have degree one in $\bds$, however in this case we also have to take into account the labeling function $\labf$. That is, two such vertices $s$ and $t$ can only be paired if they belong to the same induced disjoint union of paths $\indp_i$. % Therefore, whenever $s$ and $t$ are paired we require that $\labf(\component{\bds}(s)) = \labf(\component{\bds}(t))$.
	
	In accordance with the above discussion, we define the table entries as follows. We let $\dptable[t, (\bds, \mvc, \labf, \matlab)] = 1$ if and only if the following conditions are satisfied. %\footnote{Note that (\ref{idp:table:def:2}), (\ref{idp:table:def:3}) and the first part of (\ref{idp:table:def:5}) are equal to (\ref{lip:table:def:2}), (\ref{lip:table:def:3}) and (\ref{lip:table:def:4}) in the definition of the table entries for \textsc{Longest Induced Path}, respectively. This is because in both cases, $\bds$ is an induced disjoint union of paths.}

\begin{enumerate}

\makeatletter	
\renewcommand\labelenumi{(\roman{enumi})}
\renewcommand\theenumi{\roman{enumi}}
\renewcommand\labelenumii{(\alph{enumii})}
\renewcommand\theenumii{.\alph{enumii}}
\renewcommand\p@enumiii{\theenumi\theenumii}
\makeatother

	\item There is an induced disjoint union of paths $\indp = (\indp_1,\ldots, \indp_k)$ in $G[V_t \cup \bd(\bar{V_t})]$, such that for $i \neq j$, $\indp_i$ does not contain a vertex adjacent (in $G$) to a vertex in $\indp_j$. For each $i \in [k]$, we have that if $v_i \in \{\terminal_i, \terminall_i\} \cap V_t$, then $v_i \in \indp_i$. Furthermore, $v_i$ has degree one in $\indp_i$.\label{idp:table:def:1}
	\item \label{idp:table:def:2}
	\begin{enumerate}
		\item $E(\indp) \cap E(\crosg{t}) = E(\bds)$. \label{idp:table:def:2:bds}
		\item $\labf \colon \components(\bds) \to [k]$ is a labeling function of the connected components of $\bds$, such that for each component $C \in \components(\bds)$, $\labf(C) = i$ if and only if $C \subseteq \indp_i$. \label{idp:table:def:2:labf}
	\end{enumerate}
	\item $\mvc$ is a minimal vertex cover of $\crosg{t} - V(\bds)$ such that $V(\indp) \cap \mvc = \emptyset$. \label{idp:table:def:3}
	\item Let $\degonevert$ denote the set of vertices in $\bds[V_t]$ that have degree one in $\bds$ and let $\terminals_t = \terminals \cap V_t$.
		Then, $\matlab$ is a pairing (i.e.\ a partition into pairs) of the vertices in $\degonevert \symdiff \terminals_t$ with the following properties.\footnote{We denote by `$\symdiff$' the symmetric difference, i.e.\ $\degonevert \symdiff \terminals_t = (\degonevert \cup \terminals_t) \setminus (\degonevert \cap \terminals_t)$. $\matlab$ is a pairing on $\degonevert \symdiff \terminals_t$, since if a terminal $v_i$ is contained in $\degonevert$, it is supposed to be paired `with itself': Since $v_i$ has degree one in $\indp_i$ by (\ref{idp:table:def:1}) and is incident to an edge in $\bds$, $v_i$ cannot be paired with another vertex.} \label{idp:table:def:4}
	\begin{enumerate}
		\item $(\matchingvertex, \matchingvertexx) \in \matlab$ if and only if there is a path from $s$ to $t$ in $\indp[V_t]$. Note that this implies that $(\matchingvertex, \matchingvertexx) \in \matlab$ only if both $\matchingvertex$ and $\matchingvertexx$ belong to the same $\indp_i$ for some $i \in [k]$ and in particular only if $\matchingvertex, \matchingvertexx \in V(\labf^{-1}(i)) \cup \{\terminal_i, \terminall_i\}$.\label{idp:table:def:4:1}
		\item For each $i \in [k]$, $(\terminal_i, \terminall_i) \in \matlab$ only if $\labf^{-1}(i) = \emptyset$, i.e.\ no component of $\bds$ has label $i$.\label{idp:table:def:4:2}
	\end{enumerate}
\end{enumerate}	
	
	We now show that the answer to the problem can be obtained from inspecting the table entries stored in the root of $\dectree$.
	
	\begin{proposition}\label{prop:idp:root:corr}
		$G$ contains a set of vertex-disjoint induced paths $\idpsol = (\idpsol_1, \ldots, \idpsol_k)$, where $\idpsol_i$ is an $(\terminal_i, \terminall_i)$-path for each $i \in [k]$ and for $j \neq i$, no vertex in $\idpsol_i$ is adjacent to a vertex in $\idpsol_j$, if and only if $\dptable[r, (\emptyset, \emptyset, \emptyset, \matching_\emptyset)] = 1$, where $\matching_\emptyset = \{(\terminal_1, \terminall_1), \ldots, (\terminal_k, \terminall_k)\}$.
	\end{proposition}
	\begin{proof}
		Suppose that $\idpsol = (\idpsol_1, \ldots, \idpsol_k)$ is a solution to \textsc{Induced Disjoint Paths}, i.e.\ $\idpsol$ satisfies the conditions of the proposition. First note that since $V_r = V$ and $\bd(\bar{V_r}) = \emptyset$, $G[V_r \cup \bd(\bar{V_r})] = G$. This implies that $\idpsol$ satisfies the conditions stated in Part (\ref{idp:table:def:1}) of the definition of a table entry being set to $1$. (Clearly, each terminal $v_i \in \{\terminal_i, \terminall_i\}$ is contained in $V_r = V$ and has degree one in $\idpsol_i$.) Since $\bd(\bar{V_r}) = \emptyset$, it follows that $\crosg{r} = \emptyset$ and hence, $\bds = \emptyset$, $\mvc = \emptyset$ and $\labf = \emptyset$. Since each $\idpsol_i$ is an (induced) $(\terminal_i, \terminall_i)$-path and $\matching_\emptyset$ pairs terminals $(\terminal_i, \terminall_i)$ for each $i \in [k]$ by assumption, $\matching_\emptyset$ satisfies Part (\ref{idp:table:def:4}) of the definition of a table entry being set to $1$. Remark that Part (\ref{idp:table:def:4:2}) is satisfied since $\labf = \emptyset$ and hence $\labf^{-1}(i) = \emptyset$ for all $i \in [k]$.
		It follows that $\dptable[r, (\emptyset, \emptyset, \emptyset, \matching_\emptyset)] = 1$.
	
		Now suppose that $\dptable[r, (\emptyset, \emptyset, \emptyset, \matching_\emptyset)] = 1$. Then by Part (\ref{idp:table:def:1}) of the definition of the table entries, there is an induced disjoint union of paths $\indp = (\indp_1, \ldots, \indp_k)$ in $G[V_r \cup \bd(\bar{V_r})] = G$ such that for each $i \neq j$, no vertex in $\indp_i$ is adjacent to a vertex in $\indp_j$. Furthermore, for each $i \in [k]$, $\{\terminal_i, \terminall_i\} \subseteq V(\indp_i)$, and both $\terminal_i$ and $\terminall_i$ have degree one in $\indp_i$. Since $\matching_\emptyset$ pairs all $(\terminal_i, \terminall_i)$, Part (\ref{idp:table:def:4}) allows us to conclude that each $\indp_i$ is an $(\terminal_i, \terminall_i)$-path.
	\end{proof}

	Again, we denote by $\allbdsubsets_t$ the set of all induced unions of paths on at most $4w$ vertices in $\crosg{t}$, for $\bds \in \allbdsubsets_t$ and by $\allminvertexcovers_{t, \bds}$ the set of all minimal vertex covers of $\crosg{t} - V(\bds)$. We let $\alllabf_{\bds}$ denote the set of all labeling functions of the connected components of $\bds$ and for $\labf \in \alllabf_{\bds}$ by $\allmatchings_{t, \bds, \labf}$ the set of all pairings in accordance with Part (\ref{idp:table:def:4}) of the table definition.

	\begin{proposition}\label{prop:idp:table:size}
		For each $t \in V(\dectree)$, there are at most $n^{\cO(w)}$ table entries in $\dptable_t$ and they can be enumerated in time $n^{\cO(w)}$.
	\end{proposition}
	\begin{proof}
		By the proof of Proposition \ref{prop:lip:table:size}, we know that $\card{\allbdsubsets_t} \le n^{\cO(w)}$ and $\card{\allminvertexcovers_{t, \bds}} \le \cO(n^{4w})$.  Since Lemma \ref{lem:disjoint:paths:size} also bounds the number of connected components in $\bds \in \allbdsubsets_t$ by $\cO(w)$, we can conclude that $\card{\alllabf_{\bds}} \le k^{\cO(w)}$, since then for each $\labf \in \alllabf_{\bds}$, $\labf^{-1}(i) \neq \emptyset$ for at most $\cO(w)$ values of $i \in [k]$. 
	The same reasoning can be applied to count $\card{\allmatchings_{t, \bds, \labf}}$, since each $j \in [k]$ such that $\labf^{-1}(j) = \emptyset$ allows no further choices for pairings in $\allmatchings_{t, \bds, \labf}$: Either both $\terminal_j$ and $\terminall_j$ are contained in $V_t$ and we pair them, or neither of them is contained in $V_t$ and we disregard them. We can conclude that $\card{\allmatchings_{t, \bds, \labf}} \le w^{\cO(w)}$. To summarize, there are at most
		\[
			\cO(n^{4w}) \cdot n^{\cO(w)} \cdot k^{\cO(w)} \cdot w^{\cO(w)} = n^{\cO(w)}
		\]
	entries in $\dptable_t$ for each $t \in V(\dectree)$. Note that even if $k = \cO(n)$, we obtain a bound of $n^{\cO(w)}$ on the number of table entries stored at each node $t \in V(T)$. The time bound on the enumeration of the indices follows from the same argument given in the proof of Proposition \ref{prop:lip:table:size}.
	\end{proof}
	
	We now explain how the table entries $\dptable_t$ are computed for each $t \in V(\dectree)$.
	
	\textbf{Leaves of $\dectree$.} Let $t \in V(\dectree)$ be a leaf of $\dectree$ and $v = \decf^{-1}(t)$. Any nonempty intersection of an induced disjoint union of paths with $\crosg{t}$ is either a single edge using $v$ or a path on two edges having $v$ as the middle vertex. Hence, all nonempty $\bds$ such that the corresponding table entry is to $1$ are either a single edge using $v$, and we denote this set by $\allbdsubsets_{t, v}^1$ or a path on two edges with $v$ as the middle vertex and we denote this set by $\allbdsubsets_{t, v}^2$.
	
	Suppose $\bds \in \allbdsubsets_{t, v}^1$. Since $\crosg{t} - V(\bds)$ has no edges, $\mvc = \emptyset$. If a solution intersects $\bds$, this means that $v$ is a terminal vertex, i.e.\ $v \in \{\terminal_{i^*}, \terminall_{i^*}\}$ for some $i^* \in [k]$. Hence, (with a slight abuse of notation) $\labf(\bds) = i^*$ and since $v$ is the only degree one vertex in $\bds[V_t]$, $\matlab = \emptyset$.
	
	Now suppose $\bds \in \allbdsubsets_{t, v}^2$. Again, $\crosg{t} - V(\bds)$ has no edges, so $\mvc = \emptyset$. Let $w_1, w_2$ be the vertices in $\bds$ other than $v$. We distinguish the following cases.
	\begin{description}
		\item[Case L1 (no terminal).] If neither of $w_1$ and $w_2$ is a terminal vertex, then $\bds$ can be a subgraph of any $(\terminal_i, \terminall_i)$-path, so we allow all choices of $\labf(\bds) = i$ for $i \in [k]$.
		\item[Case L2 (one terminal).] If precisely one of $w_1$ and $w_2$ is a terminal vertex for some $i^* \in [k]$, then the only choice for $\labf$ is such that $\labf(\bds) = i^*$.
		\item[Case L3 (two terminals, same path).] If $\{w_1, w_2\} = \{\terminal_{i^*}, \terminall_{i^*}\}$ for some $i^* \in [k]$, then the only choice for $\labf$ is such that $\labf(\bds) = i^*$. Note that if $w_1 \in \{\terminal_i, \terminall_i\}$ and $w_2 \in \{\terminal_j, \terminall_j\}$ for $i \neq j$, then $\bds$ is not a valid partial solution.
	\end{description}
	In all of the above cases, there is no degree one vertex in $\bds[V_t]$, so $\matlab = \emptyset$.
	
	Eventually, we have to consider the cases when no edge of $\crosg{t}$ is used in a partial solution. Note that this is only possible if $v$ is not a terminal vertex. The choices for the minimal vertex covers are then $\{v\}$ and $N(v)$ and since $\bds = \emptyset$, it follows that $\labf = \emptyset$ and $\matlab = \emptyset$.
	
	To summarize, we set the table entries in $\dptable_t$ as follows.
	\begin{align*}
		\dptable[t, (\bds, \mvc, \labf, \matlab)] = 
			\left\{\begin{array}{ll}
				1, &\mbox{if } \bds \in \allbdsubsets_{t, v}^1, \mvc = \emptyset, \labf(\bds) = i^*, \matlab = \emptyset \mbox{ and } v \in \{\terminal_{i^*}, \terminall_{i^*}\} \\
				1, &\mbox{if } \bds \in \allbdsubsets_{t, v}^2, \mvc = \emptyset, \labf(\bds) = i \mbox{ for some } i \in [k], \matlab = \emptyset \\
					&\mbox{and } \{w_1, w_2\} \cap \terminals = \emptyset \mbox{ (Case L1)} \\
				1, &\mbox{if } \bds \in \allbdsubsets_{t, v}^2, \mvc = \emptyset, \labf(\bds) = i^*, \matlab = \emptyset \\
					&\mbox{and } w_{j_1} \in \{\terminal_{i^*}, \terminall_{i^*}\}, w_{j_2} \notin \terminals \mbox{ where } \{j_1, j_2\} = [2] \mbox{ (Case L2)} \\
				1, &\mbox{if } \bds \in \allbdsubsets_{t, v}^2, \mvc = \emptyset, \labf(\bds) = i^*, \matlab = \emptyset \\
					&\mbox{and } \{w_1, w_2\} = \{\terminal_{i^*}, \terminall_{i^*}\} \mbox{ (Case L3)} \\
				1, &\mbox{if } \bds = \emptyset, \mvc \in \{\{v\}, N(v)\}, \labf = \emptyset, \matlab = \emptyset \mbox{ and } v \notin \terminals \\
				0, &\mbox{otherwise}
			\end{array}\right.		
	\end{align*}
	
	\textbf{Internal nodes of $\dectree$.} We argue that for each internal node $t \in V(\dectree)$ and each index $\indexshort_t \defeq (\bds, \mvc, \labf, \matlab) \in \allbdsubsets_t \times \allminvertexcovers_{t, \bds} \times \alllabf_{\bds} \times \allmatchings_{t, \bds, \labf}$, the value of the table entry $\dptable[\indexshort_t]$ can be computed in complete analogy with the algorithm for \textsc{Longest Induced Path} (LIP).  In particular, it can be viewed as performing for each $i \in [k]$ the steps presented in the algorithm of LIP. Some attention must be devoted to guaranteeing that there is no edge between two components labeled $i$ and $j$ (where $i \neq j$). However, as outlined below, partial solutions to the two problems represented by their respective indices have almost exactly the same structure and in that sense, the routine presented for LIP already captures this requirement. More precisely, the only way in which unwanted edges could be added during the join in \textsc{Induced Disjoint Paths} (IDP) is if they cross the boundary. By definition, there will never be any unwanted edges in the intersection of a partial solution with the boundary in both problems. Subsequently, the only place where they can appear is between intermediate vertices of a combined partial solution. In the join of LIP, we explicitly ensure that no edges between intermediate vertices appear (via the minimal vertex cover) and in the same way we can enforce in IDP that no edges appear between intermediate vertices of components labeled $i$ and $j$ (for $i \neq j$) when combining two partial solutions.
	
	We now argue in detail the similarity between the partial solution based on the definitions of the table entries for the respective problems. Throughout the following, we refer to the conditions stated in the definition of the table entries for \textsc{Longest Induced Path} by LIP($\cdot$) and for \textsc{Induced Disjoint Paths} by IDP($\cdot$).
	
	First, observe that in both cases, the intersection of a solution with the graph $G[V_t \cup \bd(\bar{V_t})]$ is an induced disjoint union of paths $\indp$. The key difference between the two problems is that in LIP we are interested in the solution size, whereas in IDP we are not and that in IDP we additionally have to take	into account to which $(\terminal_i, \terminall_i)$-path the components of $\indp$ belong. % in particular that we have explicit knowledge of the endpoints of the induced paths we are looking for. 
	Note that aside from the size constraint, LIP(\ref{lip:table:def:1}) and IDP(\ref{idp:table:def:1}) express precisely the same thing if $k=1$ in IDP and the entry $j$ in LIP takes the role of the fixed terminals in IDP.
	
	LIP(\ref{lip:table:def:2}) and IDP(\ref{idp:table:def:2:bds}) are in fact identical, in particular this means that the intersection $\bds$ of a partial solution with the crossing graph $\crosg{t}$ is the same in both problems. 
	 So to make the join procedure described for LIP work for IDP in a way that it preserves Condition IDP(\ref{idp:table:def:2}), we only have to take into account the behavior of the labeling function $\labf$ introduced in IDP(\ref{idp:table:def:2:labf}). Note that again, if $k =1$ in IDP, then LIP(\ref{lip:table:def:2}) and IDP(\ref{idp:table:def:2}) express precisely the same thing. 
	 LIP(\ref{lip:table:def:3}) and IDP(\ref{idp:table:def:3}) are as well identical, so all parts of the algorithm for LIP that ensure that this condition holds can immediately applied (and argued for) in the join of IDP.
	 
	 It remains to argue the similarity of LIP(\ref{lip:table:def:4}) and IDP(\ref{idp:table:def:4}). Since we do not know the endpoints of the induced path we are looking for in LIP, whereas in IDP we do, slightly differing languages were used to express the properties of the pairings. In LIP(\ref{lip:table:def:4}), two vertices are paired if contracting all edges in $\indp - E(\crosg{t})$ leaves an edge between the two vertices in the graph. But this happens precisely when there is a path between these two vertices in $\indp[V_t]$, the condition stated in IDP(\ref{idp:table:def:4:1}).
	The remainder of IDP(\ref{idp:table:def:4}) is devoted to including the terminals in the pairings, so the necessary modifications in the algorithm of LIP to work for IDP such that it respects IDP(\ref{idp:table:def:4}) are straightforward. Note that again if $k = 1$ in IDP, LIP(\ref{lip:table:def:4}) and IDP(\ref{idp:table:def:4}) express the same thing where $j$ takes the role of the terminals contained in $V_t$.
	
	By the above discussion, the fact that we enumerate all possible solutions in the leaves of $\dectree$, Propositions \ref{prop:idp:root:corr} and \ref{prop:idp:table:size} and in the light of the correctness (and runtime) of the join of the algorithm for \textsc{Longest Induced Path}, we have the following theorem.
	
\begin{theorem}\label{thm:induced:disjoint:paths:mim}
	There is an algorithm that given a graph $G$ on $n$ vertices, pairs of terminal vertices $(\terminal_1, \terminall_1), \ldots, (\terminal_k, \terminall_k)$ and a branch decomposition $(\dectree, \decf)$ of $G$, solves \textsc{Induced Disjoint Paths} in time $n^{\cO(w)}$, where $w$ denotes the mim-width of $(\dectree, \decf)$.
\end{theorem}

\subsection{$H$-Induced Topological Minor}\label{sub:sec:alg:ITM}
Let $G$ be a graph and $uv \in E(G)$. We call the operation of replacing the edge $uv$ by a new vertex $x$ and edges $ux$ and $xv$ the {\em edge subdivision} of $uv$. We call a graph $H$ a {\em subdivision} of $G$ if it can be obtained from $G$ by a series of edge subdivisions. We call $H$ an {\em induced topological minor} of $G$ if a subdivision of $H$ is isomorphic to an induced subgraph of $G$.
\parproblemdef
	{$H$-Induced Topological Minor/Mim-Width}
	{A graph $G$ with branch decomposition $(\dectree, \decf)$}
	{$w \defeq \mimw(\dectree, \decf)$}
	{Does $G$ contain $H$ as an induced topological minor?}

\begin{theorem}
	There is an algorithm that given a graph $G$ on $n$ vertices and a branch decomposition $(\dectree, \decf)$ of $G$, solves \textsc{$H$-Induced Topological Minor} in time $n^{\cO(w)}$, where $H$ is a fixed graph and $w$ the mim-width of $(\dectree, \decf)$.
	\end{theorem}	
	\begin{proof}
	Let $H$ be a fixed graph. To solve \textsc{$H$-Induced Topological Minor}, we have to find a map $\varphi$ from the vertices of $H$ to a set of vertices in $G$ such that there is an edge in $\{x, y\}$ if and only if there is an induced path between $\varphi(x)$ and $\varphi(y)$ in $G$ such that additionally, for each pair $P_1$, $P_2$ of such paths, no vertex in $P_1$ is adjacent to a vertex in $P_2$. We do so by guessing to which vertices in $G$ the vertices of $H$ are mapped and after some preprocessing, we run the algorithm for \textsc{Induced Disjoint Paths} with (at most) $\card{E(H)}$ pairs of terminals to find the induced paths in $G$ corresponding to the edges in $H$.
	
	\begin{figure}
		\centering
		\includegraphics[height=.15\textheight]{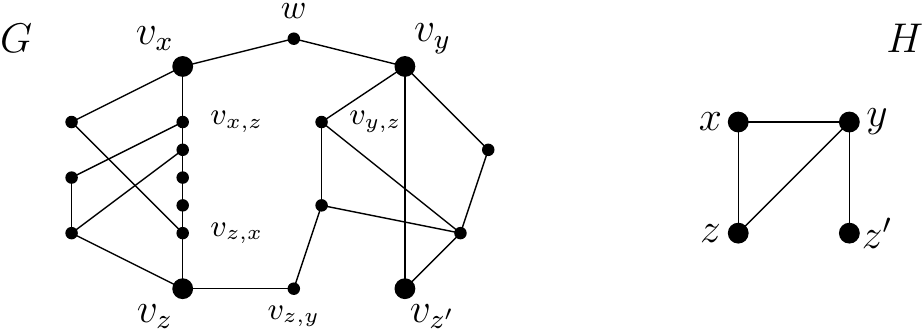}
		\caption{An example \textsc{Yes}-instance of \textsc{$H$-Induced Topological Minor}. Note that $v_{y, z'} = v_{z'}$ and $v_{z', y} = v_y$ since $\{v_y, v_{z'}\} \in E(G)$ in accordance with Step 2(\ref{hitm:step:2:2}). For each of the remaining edges in $H$ there are precisely two corresponding edges each in $G$, illustrating Step 2(\ref{hitm:step2:1}). Furthermore, $w = v_{x, y} = v_{y, x}$, illustrating the special case dealt with in Step 2(\ref{hitm:step:2:3}).}
		\label{fig:hitm:ex}
	\end{figure}	
	
	\begin{description}
		\item[Step 1 (Branching).] For each vertex $x \in V(H)$, we guess the corresponding vertex $v_x \in V(G)$ of $x$, i.e.\ $\varphi(x) = v_x$. Hence we require that for all $x, y \in V(H)$, $x = y \Leftrightarrow v_x = v_y$ and that $\deg_G(v_x) \ge \deg_H(x)$.
		We let $X \defeq \bigcup_{x \in V(H)} v_x$. Furthermore, for each $x \in V(H)$ and neighbor $y$ of $x$, we guess a distinct corresponding neighbor $v_{x, y}$ in $G$. Note that since $\deg_G(v_x) \ge \deg_H(x)$ for all $x \in V(H)$, there are sufficiently many such vertices in $G$ to choose from.
		The vertex $v_{x, y}$ is the vertex adjacent to $v_x$ in the induced path in $G$ corresponding to the edge $\{x, y\}$ in $H$. We let $Y_x \defeq \bigcup_{y \in N_H(x)} v_{x, y}$ and $Y \defeq \bigcup_{x \in V(H)} Y_x$ and we refer to the vertices in $Y$ as the {\em neighbor vertices}. We branch on each such choice of $X \cup Y$.		
		
	\item[Step 2 (Preprocessing according to $X \cup Y$).] In this step we check whether the current choice of $X$ and $Y$ is valid and prepare the current instance for the call to the algorithm of \textsc{Induced Disjoint Paths}.
		\begin{enumerate}[(a)]
			\item For each edge $\{x, y\} \in E(H)$, if $\{v_x, v_y\} \in E(G)$, then we remove the edge $\{x, y\}$ from $H$ without changing the answer to the problem: $\{v_x, v_y\}$ is the induced path in $G$ corresponding to the edge $\{x, y\}$. 
				If on the other hand, $\{x, y\} \notin E(H)$ but $\{v_x, v_y\} \in E(G)$, then we discard this choice of $X \cup Y$. \label{hitm:step:2:2}
			\item We check whether for each edge $\{x, y\} \in E(H)$ there are precisely two corresponding edges, namely $\{v_x, v_{x, y}\}$ and $\{v_y, v_{y, x}\}$ in $G$ and that $G[X \cup Y]$ induces no further edges. In particular, if $G[X \cup Y]$ does contain any additional edges, we discard this choice of $X$ and $Y$. \label{hitm:step2:1}
			\item If there is a vertex $w \in V(G)$ which is chosen more than once as a neighbor vertex, i.e.\ there exists a set $A \subseteq V(H)$ such that $w \in \bigcap_{a \in A} Y_a$ with $|A| > 1$, then we proceed only if $A = \{x, y\}$ for some $x, y \in V(H)$ and such that $w = v_{x, y} = v_{y, x}$. In that case, we remove the edge $\{x, y\}$ from $H$ without changing the answer to the problem: The set $\{v_x, w, v_y\}$ induces the path in $G$ corresponding to the edge $\{x, y\}$. Furthermore, we remove the vertices in $N(w) \setminus X$ from $G$, since these vertices cannot be used by any other path corresponding to an edge in $H$. \label{hitm:step:2:3}
		\end{enumerate}
		
		\item[Step 3 (Execution of IDP-algorithm).] We run the algorithm for \textsc{Induced Disjoint paths} on $G - X$ with pairs of terminals $(v_{x, y}, v_{y, x})_{\{x, y\} \in E(H)}$. (Note that some edges of $H$ might have been removed in Steps 2(\ref{hitm:step:2:2}) and 2(\ref{hitm:step:2:3}) and that some vertices of $G$ might have been removed in Step 2(\ref{hitm:step:2:3}).)
		
		\item[Step 4 (Return).] We let the algorithm return \textsc{Yes} if at least one run of \textsc{Induced Disjoint Paths} in Step 3 returned \textsc{Yes}, and \textsc{No} otherwise.
	\end{description}
	
	We now analyze the runtime of the algorithm. In Step 1, we branch in $n^{\cO(\card{E(H)})}$ ways, the number of choices for the sets $X$ and $Y$. The checks in Step 2 can be performed in time polynomial in $n$ and each execution of the algorithm for \textsc{Induced Disjoint Paths} in Step 3 takes time $n^{\cO(w)}$ by Theorem \ref{thm:induced:disjoint:paths:mim}. So the total runtime of the algorithm is $n^{\cO(\card{E(H)})} \cdot n^{\cO(w)} = n^{\cO(w)}$, as $H$ is fixed.
	\end{proof}

% \section{Conclusion}\label{sec:conclusion}

\bibliography{fvsmim}
\end{document}